\documentclass[12pt,draftcls,onecolumn]{IEEEtran}
\IEEEoverridecommandlockouts

\usepackage{amsmath,amssymb,amsfonts}
\usepackage{graphicx}
\usepackage{textcomp}
\usepackage{color}
\usepackage{graphicx}      
\usepackage{cite}
\usepackage{bm}
\usepackage{amsthm}
\usepackage{epsfig}
\usepackage{epstopdf}
\usepackage{float}
\usepackage{mathrsfs}
\usepackage{xtab}
\usepackage{color}
\usepackage{ifthen}
\usepackage{setspace}
\usepackage{algorithm}
\usepackage{algorithmic}
\usepackage{subcaption}
\usepackage[dvipsnames]{xcolor}

\allowdisplaybreaks[4]
\allowdisplaybreaks
\newtheorem{rem}{Remark}
\newtheorem{lem}{Lemma}
\newtheorem{theo}{Theorem}
\newtheorem{coro}{Corollary}
\newtheorem{ass}{Assumption}
\newtheorem{prop}{Proposition}
\newtheorem{defi}{Definition}

\ifCLASSINFOpdf
\else
\fi
\begin{document}
\allowdisplaybreaks

\title{Control Reconfiguration of Dynamical Systems for Improved Performance via Reverse- and Forward-engineering}
%
%
%

\author{Han~Shu,
        Xuan~Zhang*,
        Na~Li,
and~Antonis~Papachristodoulou~\IEEEmembership{Fellow,~IEEE}
\thanks{This work was supported by Tsinghua-Berkeley Shenzhen Institute Research Start-Up Funding. H. Shu and X. Zhang are with the Smart Grid and Renewable Energy Laboratory, Tsinghua-Berkeley Shenzhen Institute, Shenzhen, Guangdong 518055, China (email: hs2226@cornell.edu, xuanzhang@sz.tsinghua.edu.cn).  \textit{Corresponding Author: X. Zhang.}}
\thanks{N. Li is with the School of Engineering and Applied Sciences, Harvard University, Cambridge, MA, USA (email: nali@seas.harvard.edu).}
\thanks{ A. Papachristodoulou is with the Department of Engineering Science, University of Oxford, Oxford, UK (email: antonis@eng.ox.ac.uk).}}

\maketitle

\begin{abstract}
This paper presents a control reconfiguration approach to improve the performance of two classes of dynamical systems. Motivated by recent research on re-engineering cyber-physical systems, we propose a three-step control retrofit procedure. First, we reverse-engineer a dynamical system to dig out an optimization problem it actually solves. Second, we forward-engineer the system by applying a corresponding faster algorithm to solve this optimization problem. Finally, by comparing the original and accelerated dynamics, we obtain the implementation of the redesigned part (the extra dynamics). As a result, the convergence rate/speed or transient behavior of the given system can be improved while the system control structure is maintained. Internet congestion control and distributed proportional-integral (PI) control, as applications in the two different classes of target systems, are used to show the effectiveness of the proposed approach.
\end{abstract}


%
\IEEEpeerreviewmaketitle

\begin{spacing}{1.0}
\section{Introduction}

\IEEEPARstart{C}{yber}-physical systems (CPSs) integrate, coordinate and monitor the operations of both a physical process and the cyber world~\cite{rajkumar2010cyber}. They 
are decisive in supporting fundamental infrastructures and smart applications including automotive systems, transportation systems, smart grids, etc. However, although those CPSs were advanced at the time when being constructed, they can be either economical inefficient or energy inefficient from today's viewpoint. 
For example, aging electricity distribution infrastructures are becoming less reliable and less efficient~\cite{kim2013overview}. Also, there are societal and industrial needs for better control of CPSs. For example, the increasing penetration of renewable energy poses threats to the reliability of power grids: it is essential to design better control to quickly attenuate large fluctuations caused by those energy sources. Furthermore, autonomous vehicles and mobile robots, being operated in an uncertain environment without complete information, require better control for more safety and reliability. 

Usually, better control can be achieved from two perspectives. One way is rebuilding a new controller for the whole system. For example,~\cite{deveci2016performance} shows that by incorporating numerical modeling and simulation to design the controller for a photovoltaic system, system performance and robustness can be improved. The other way is to redesign the existing controller by adding extra dynamics while maintaining the control structure of the existing controller, i.e., if the original control is centralized/distributed, then the redesigned control is also centralized/distributed. This is referred to as control reconfiguration and can be simply realized, for example, by using additional information produced by newly added sensors, without affecting the structure of the controlled system. For instance,~\cite{zhang2014improving} proposes a modification approach based on a penalty method for improving the performance and robustness to delays of Internet congestion control protocols, and \cite{nevsic2005lyapunov} redesigns controllers by adding extra terms obtained based on continuous-time systems and Lyapunov functions to enhance stability and robustness. In general, there are tradeoffs between these two methods: rebuilding a new controller can achieve a better result by implementing the state-of-the-art sensing, communication and computing technologies but with the expense of complexity and high investment, while redesigning the existing controller can be easier and  more convenient though the effect may not be as good as rebuilding. 

Instead of designing a new controller, this paper focuses on modifying the existing control from an optimization perspective to improve the performance of the whole system. Inspired by recent research on re-engineering typical CPSs~\cite{Low1999Optimal,li2015connecting,chen2016reverse,zhang2018distributed,zhou2020reverse}, this work extends to study control of general CPSs based on a reverse- and forward-engineering framework, which serves as a tool to bridge the gap between engineering CPS applications and existing theoretical results on optimization. As will be shown later, the proposed framework shows great potential to handle large-scale system cases.

The idea of redesign using a reverse- and forward-engineering framework for optimality has been introduced for over ten years~\cite{chiang2007layering}. From existing protocols designed based on engineering instincts, utility functions are implicitly determined and can be extracted via reverse-engineering. Other works consider various congestion control protocols as distributed algorithms for solving network utility maximization problems~\cite{kelly1998rate,Kunniyur2000End,Low1999Optimal}. Based on the insights obtained from reverse-engineering, forward-engineering systematically improves the protocols~\cite{chiang2007layering,Low1999Optimal}. Recently, inspired by this idea,~\cite{li2015connecting} connects automatic generation control (AGC) and economic dispatch by reverse-engineering AGC to improve power system economic efficiency, and ~\cite{zhang2015distributed,zhang2018distributed} develop a reverse- and forward-engineering framework to redesign control for improved efficiency, achieving optimal steady-state performance in network systems.

Nevertheless, little attention has been paid to improve system performance in terms of convergence rates/speeds and transient behaviors. This paper utilizes the reverse- and forward-engineering framework, together with several acceleration techniques, to systematically improve the performance of two classes of dynamical systems in discrete time. These systems include but are not limited to existing protocols and controlled systems, e.g., Internet congestion control, distributed proportional-integral (PI) control.

The main contribution of this paper is twofold.
\begin{itemize}
\item  A control reconfiguration approach is proposed to systematically improve the performance of two classes of general dynamical systems while maintaining the original control structure. This is realized by designing extra dynamics and adding it to the original control based on a reverse- and forward-engineering framework. 

\item 
Two standard acceleration methods are utilized in each class of dynamical systems for control reconfiguration and are theoretically proven to exhibit improved convergence rates.



\end{itemize}

The rest of this paper is organized as follows. Section~\ref{se:prelim} presents Class-$\mathcal{O}$\footnotemark[1] and Class-$\mathcal{S}$\footnotemark[1] as our target systems, together with two motivating examples. Section~\ref{se:redesign} presents the reconfiguration steps for systems in Class-$\mathcal{O}$ and Class-$\mathcal{S}$ and analyzes the convergence rates of the original and redesigned systems. Section~\ref{se:simu} provides simulation results of the motivating examples to illustrate the effectiveness of the reconfiguration framework. Section~\ref{se:conclusion} concludes the paper.
\footnotetext[1]{Notation $\mathcal{O}$ stands for \emph{optimization algorithms} in contrast to $\mathcal{S}$ for \emph{saddle-point algorithms}.}


\noindent\textbf{Notations:} Throughout this paper, we use upper case roman letters to denote matrices. 
Let $\text{diag}\{\star\}$ denote the diagonal matrix with corresponding entries $\star$ on the main diagonal. Let $A \succeq 0$ ($A \succ 0$) denote that a square matrix is positive semi-definite (positive definite). For two symmetric matrices $A_1$ and $A_2$, notation $A_1 \preceq A_2$ implies that $A_2 - A_1$ is positive semi-definite. Let $\left\langle { \cdot {\rm{ }},{\rm{ }} \cdot } \right\rangle$ represent the Euclidean inner product, and let $\left\| {{\rm{ }} \cdot {\rm{ }}} \right\|$ denote the Euclidean norm for vectors and the spectrum norm for matrices. Denote by $\mathbb{Z}^+$ the set of positive integers, by $\mathbb{R}^n$ the $n$-dimensional Euclidean space and by $A \in \mathbb{R}^{m \times n}$ an $m \times n$ real matrix. Let $x^T$ and $\nabla f$ denote the transpose of $x$ and the gradient of $f$ (as a column vector). 
Let ${\sigma _{\max }(B)}$/${\sigma _{\min }(B) }$ denote the maximum/minimum singular values of $B$. Let ${\lambda(A)}$ denote the eigenvalues of a square matrix $A$. Denote by $\bf{H}$$_f$ the Hessian matrix of $f$ with elements $\bf{H}$$_{i,j}=\frac{\partial^2 f}{\partial x_i \partial x_j}$. 

\section{Preliminaries} \label{se:prelim}
In this section, we introduce two special classes of dynamical systems as our focus, i.e., Class-$\mathcal{O}$ and Class-$\mathcal{S}$, and present one typical example in each class as the motivating applications. In particular, we show what kinds of conditions are required to determine whether or not a system belongs to these classes (when reverse-engineering works)~\cite{zhang2018distributed}, mainly for the discrete-time linear case. 

Reverse-engineering seeks a proper optimization problem inherently from given dynamical equations. Different from the previous optimization-to-algorithm framework, it generates a reverse flow, i.e., algorithm-to-optimization.

Consider a linear time-invariant (LTI) system
\begin{align}
x_{k+1}=Ax_k+Cw \label{eq:LTI-sys}
\end{align}where $x_k\in\mathbb{R}^{n}$ is the state vector, $A \in \mathbb{R}^{n \times n}$, $C \in \mathbb{R}^{n \times p}$ and $w\in\mathbb{R}^{p}$ is the exogenous input, e.g., disturbances. Note that $w$ can be constant or time-varying. In general, any given discrete-time LTI closed-loop system with either static feedback or dynamic feedback can be rearranged to fit~\eqref{eq:LTI-sys}.

\subsection{Target Systems: Class-$\mathcal{O}$}

\noindent\textbf{Class-$\mathcal{O}$}:
System~\eqref{eq:LTI-sys} belongs to Class-$\mathcal{O}$ if there exists a convex function $f(x): \mathbb{R}^{n}\rightarrow\mathbb{R}$ and a positive definite matrix $P \in \mathbb{R}^{n \times n}$, such that the set $\{\nabla f(x)=\textbf{0}\}$ is nonempty and system~\eqref{eq:LTI-sys} is a gradient descent algorithm to solve an unconstrained convex optimization problem $\min_{x} f(x)$, i.e., $x_{k+1} = x_k - P \nabla f(x) |_{x=x_k}$.

For linear systems in Class-$\mathcal{O}$, the associated $f$ must be a convex quadratic function, i.e., $f=\frac{1}{2}x^TQx+x^TR(Cw)+s(w)$ for some matrices $Q \succeq 0$ and $R$, and $s(w)$ stands for terms only related to  $w$. Note that $f$ is regarded as a function of the variable $x$ and $w$ is treated as a parameter. Therefore, system~\eqref{eq:LTI-sys} belongs to Class-$\mathcal{O}$ if and only if there exist $P \succ 0 $ and $Q \succeq 0$ such that $A=I_n-PQ$ and $R=-P^{-1}$ hold. This results in the following theorem~\cite{zhang2018distributed}. 

\begin{theo}\label{theo:class-O}
Let $w$ be constant in~\eqref{eq:LTI-sys} and the set $\{(A-I_n)x+Cw=\mathbf{0}\}$ be nonempty. System~\eqref{eq:LTI-sys} belongs to Class-$\mathcal{O}$ if and only if system~\eqref{eq:LTI-sys} is marginally or asymptotically stable\footnote[2]{All the eigenvalues of $A$ have negative or zero real parts and all the Jordan blocks corresponding to eigenvalues with zero real parts are $1 \times 1$.}, all the eigenvalues of $I_n-A$ are real numbers and $I_n-A$ is diagonalizable. 
\end{theo}


\subsection{Target Systems: Class-$\mathcal{S}$}


\noindent\textbf{Class-$\mathcal{S}$}: System~\eqref{eq:LTI-sys} belongs to Class-$\mathcal{S}$ if there exists a function $f(x^{(1)},x^{(2)}):\mathbb{R}^{n}\rightarrow \mathbb{R}$ and positive definite matrices $P_{x^{(1)}}$\footnotemark[3], $P_{x^{(2)}}$\footnotemark[3] such that ${\bf{H}}_{f(x^{(1)})}\footnotemark[3] = \mathbf{0}$, ${\bf{H}}_{f(x^{(2)})} \succeq 0$, the saddle-point set $\{\nabla f(x)=\mathbf{0}\}$ is nonempty, and~\eqref{eq:LTI-sys} is a primal-dual gradient algorithm to solve $\max_{x^{(1)}} \min_{x^{(2)}} f$, i.e., $x_{k+1}=x_k + \text{diag}\{P_{x^{(1)}}, -P_{x^{(2)}}\}\nabla f|_{x=x_k}$. \footnotetext[3]{Here $P_{x^{(1)}}$/$P_{x^{(2)}}$ is a positive definite matrix with dimension in accordance with $x^{(1)}$/ $x^{(2)}$; $f(x^{(1)})$ means $f$ is regarded as a function of the variable $x^{(1)}$ and $x^{(2)}$ is treated as a parameter.}
\vspace{6pt}

In the above definition, state $x$ is partitioned into $x^{(1)}$ and $x^{(2)}$, and $f$ is linear in $x^{(1)}$ (similar to that the Lagrangian function of a constrained optimization problem is linear in the dual variables). Accordingly, we rearrange system~\eqref{eq:LTI-sys} as
\begin{align}\label{equ:linearsystemsplit}
\underbrace {\left[ {\begin{array}{*{20}{c}}
{{x_{k+1}^{(1)}}}\\
{{x_{k+1}^{(2)}}}
\end{array}} \right]}_{x_{k+1}} = \underbrace {\left[ {\begin{array}{*{20}{c}}
{{A_{11}}}&{{A_{12}}}\\
{{A_{21}}}&{{A_{22}}}
\end{array}} \right]}_A \underbrace {\left[ {\begin{array}{*{20}{c}}
{x_k^{(1)}}\\
{x_k^{(2)}}
\end{array}} \right]}_{x_k}+Cw
\end{align}
where $x_k^{(1)}\in\mathbb{R}^{n_1}$, $x_k^{(2)}\in\mathbb{R}^{n_2}$ and $n_1+n_2=n$. For linear systems in Class-$\mathcal{S}$, the associated function $f$ must be a convex quadratic function, i.e.,
\begin{align}\label{equ:classSopt}
f=\frac{1}{2}
x^T\underbrace{\left[ {\begin{array}{*{20}{c}}
{{Q_{11}}}&{{Q_{12}}}\\
{Q_{12}^T}&{{Q_{22}}}
\end{array}} \right]}_Qx+x^TR(Cw)+s(w)
\end{align}
 where $Q_{11}=\mathbf{0}\in\mathbb{R}^{n_1\times n_1}$, $Q_{22}\in\mathbb{R}^{n_2\times n_2}$ is symmetric and positive semi-definite ($f$ is  linear in $x^{(1)}$ and convex in $x^{(2)}$), and $Q_{12}\in\mathbb{R}^{n_1\times n_2}$. According to the definition of Class-$\mathcal{S}$, the following theorem is obtained.

\begin{theo}\label{theo:class-S}
Let $w$ be constant in~\eqref{equ:linearsystemsplit} and the set $\{(A-I_n)x+Cw=\mathbf{0}\} $ be nonempty. System~\eqref{equ:linearsystemsplit} belongs to Class-$\mathcal{S}$ if and only if the following conditions are satisfied: (i)  system~\eqref{equ:linearsystemsplit} is marginally or asymptotically stable; (ii) the eigenvalues of $A_{11}-I_{n_1}$ and $A_{22}-I_{n_2}$ are non-positive real; (iii) $A_{11}-I_{n_1}$ and $A_{22}-I_{n_2}$ are diagonalizable with the diagonal canonical forms given by $A_{11}-I_{n_1}=J_1\Lambda_1J_1^{-1}, J_1\in\mathbb{R}^{n_1\times n_1}, A_{22}-I_{n_2}=J_2\Lambda_2J_2^{-1}, J_2\in\mathbb{R}^{n_2\times n_2}$, and there exist $V_1, V_2$ that
\begin{align}\label{equ:SDPfeasibility}
&(J_1^{-1})^TV_1J_1^{-1}A_{12}+A_{21}^T(J_2^{-1})^TV_2J_2^{-1}=\mathbf{0} \nonumber\\
&V_1\Lambda_1=\Lambda_1V_1,\text{ }V_2\Lambda_2=\Lambda_2V_2 \nonumber\\
&V_1\prec0,\text{ }V_2\prec0.
\end{align}
\end{theo}

Note that there could be multiple optimization problems corresponding to the same dynamical system~\eqref{eq:LTI-sys} in either Class-$\mathcal{O}$ or Class-$\mathcal{S}$. The derivation procedure of those problems can be found in~\cite{zhang2018distributed}.

\subsection{Motivating Applications} \label{se:application}
Many existing systems fall into the two classes. For example, consensus in multi-agent systems, primal Internet congestion control protocols and heating, ventilation and air-conditioning (HVAC) system control are in Class-$\mathcal{O}$ (see~\cite{myICCPS} for a detailed description) while primal-dual Internet congestion control protocols, distributed PI control for single integrators and frequency control in power systems are in Class-$\mathcal{S}$. Here we present one typical example for each class.

\subsubsection{Internet Congestion Control} \label{se:ICC}
Internet congestion control regulates the data transfer and efficient bandwidth sharing between sources and links. 
Here we consider a standard primal congestion control algorithm~\cite{Srikant2004The}, which belongs to Class-$\mathcal{O}$:
\begin{subequations}\label{eq:ICC_O}
\begin{align}
&x_i(k+1)=x_i(k)+\varepsilon k_{x_i}(U_i^{\prime}(x_i(k))-q_i(k)) \\
&q_i(k)=\sum_{l=1}^M R_{li}p_l(k) \\
&p_l(k)=f_l(y_l(k)) \\
&y_l(k)=\sum_{i=1}^NR_{li}x_i(k)
\end{align}
\end{subequations}where $\varepsilon$ is the step size, $U_i(x_i)$ is the utility function of user $i$, which is assumed to be a continuously differentiable, monotonically increasing, strictly concave function of the transmission rate $x_i$. $R$ is a routing matrix describing the interconnection where $R_{li} = 1$ if user $i$ uses link $l$; otherwise $R_{li} = 0$.
Also, $k_{x_i}>0$ is the rate gain, $p_l>0$ is the link price, $q_i$ is the aggregate price along user $i$'s path, $y_l$ is the total flow through link $l$, $f_l(y_l)$ is a barrier/penalty function satisfying $f_l(y_l),f_l^\prime(y_l)>0$ (usually, $f_l(y_l)$ forces the satisfaction of the inequality constraint $y_l \le c_l$ where $c_l>0$ is the link capacity of link $l$), $N$ and $M$ are the numbers of sources and links. 
The above dynamics can be rearranged in a vector form as
\begin{align}
x(k+1)=x(k) + \varepsilon \text{diag}\{k_{x_i}\} \left( U^{\prime} (x(k) ) - R^Tf(Rx(k)) \right)   \label{eq:network_matrix}
\end{align}
where $x(k), U^{\prime} (x(k) )\in \mathbb{R}^N, f(Rx(k))\in \mathbb{R}^M$ are the corresponding vector forms of $x_i(k)$, $U_i^{\prime} (x_i(k) )$, $f_l (\sum\limits_{i = 1}^N {R_{li} x_i(k) } )$.

Suppose $U_i(x_i)$ is quadratic and $f_l(y_l)$ is linear, i.e., $U(x)=-\frac{1}{2}Q \text{diag}\{x_i\}x+R_1x+s_1$ and $f(Rx)=R_2Rx+s_2$ where $Q,R_1,R_2$ are diagonal, positive definite constant matrices, and $s_1,s_2$ are constant vectors. Then~\eqref{eq:network_matrix} can be rearranged as 
\begin{align} 
 &x(k+1)=x(k) - \varepsilon \text{diag}\{k_{x_i}\}(Qx(k)+R^TR_2Rx(k))+\cdots \nonumber\\
 &=\big(I_n-\varepsilon \text{diag}\{k_{x_i}\}(Q+R^TR_2R)\big)x(k)+\cdots \label{eq:ICC_matrix}
\end{align}
where $\cdots$ denotes the remaining constant terms. Compared with~\eqref{eq:LTI-sys}, it is straightforward to notice that $I_n-A$ in~\eqref{eq:ICC_matrix} is $\varepsilon \text{diag}\{k_{x_i}\}(Q+R^TR_2R)$. 
It satisfies the conditions listed in Theorem~\ref{theo:class-O} and thus, the above dynamics~\eqref{eq:ICC_O} or \eqref{eq:network_matrix} can be reverse-engineered as a gradient descent algorithm to solve
\begin{align} \label{eq:opt_net}
\min_{x_i\ge0}-\sum_{i=1}^NU_i(x_i)+\sum_{l=1}^M\int_0^{\sum_{i=1}^NR_{li}x_i}f_l(\beta)d\beta.
\end{align}
\indent In general, the reverse-engineering from~\eqref{eq:network_matrix} to the above optimization problem always works if $U_i(x_i)$ is concave. 
It is of interest to redesign the primal congestion control algorithm~\eqref{eq:network_matrix} for a faster convergence speed to improve data transfer.

\subsubsection{Distributed PI Control for Single Integrator Dynamics} \label{se:PI}
In the following, we propose distributed PI control for single integrators as a motivating example that belongs to Class-$\mathcal{S}$. Consider $n$ agents with single integrator dynamics
\begin{align*}
 \dot y_i  &= d_i  + u_i 
\end{align*}
where $d_i$ is a constant disturbance and $u_i$ is the control input given by
\begin{align} \label{eq:integrator_u}
u_i  =  - \sum\limits_{j \in {\rm \mathcal{N}}_i } {\left( {\rho_2(y_i  - y_j ) + \rho_1 \int\limits_0^t {\big( {y_i \left( \tau  \right) - y_j \left( \tau  \right)} \big)d\tau } } \right)} - \delta \left( {y_i  - y_i(0)} \right)     
\end{align}
where $\rho_1, \rho_2, \delta$ are positive constant parameters, $y_i(0)$ is the initial condition and $\mathcal{N}_i$ is the set of neighbors of agent $i$. This controller drives agents to reach consensus under static disturbances and any initial condition according to Theorem 6 in~\cite{andreasson2014distributed}. Introduce the integral action $\dot z_i = y_i - y_n$ and rearrange the dynamics after discretizing in a vector form as 
\begin{align*} 
\underbrace {\left[ {\begin{array}{*{20}c}
   {z(k + 1)}  \\
   {y(k + 1)}  \\
\end{array}} \right]}_{x(k + 1)} =& \underbrace {\left[ {\begin{array}{*{20}c}
   {I_{n-1} } & {\varepsilon_2 D}  \\
   { - \varepsilon_1 \rho_1 \tilde{L}} & {I_n  - \varepsilon_1 \rho_2 L - \varepsilon_1 \delta I_n }  \\
\end{array}} \right]}_A\underbrace {\left[ {\begin{array}{*{20}c}
   {z(k)}  \\
   {y(k)}  \\
\end{array}} \right]}_{x(k)} \\  \nonumber
&+ \underbrace {\left[ {\begin{array}{*{20}c}
   \textbf{0}  \\
   {\varepsilon_1 \big(d+\delta y(0)\big)}  \\
\end{array}} \right]}_{Cw}
\end{align*}
where $z(k) \in \mathbb{R}^{n-1}$, $y(k), d \in \mathbb{R}^{n}$, $x(k) \in \mathbb{R}^{2n-1}$, $\varepsilon_1$ and $\varepsilon_2$ are the step sizes and $D=[I_{n-1} \quad -\textbf{1}]$ ($\textbf{1}$ is a column vector of ones). $L \in \mathbb{R}^{n \times n}$ is the Laplacian of the connected agent network and $\tilde{L} \in \mathbb{R}^{n \times (n-1)}$ is obtained after removing the $n$th column of $L$. Compared with~\eqref{equ:linearsystemsplit}, it is straightforward to notice that $A_{11}-I_{n_1}$ is $\textbf{0}$ and $A_{22}-I_{n_2}$ is $- \varepsilon_1 \rho_2 L - \varepsilon_1 \delta I_n$. They satisfy the conditions listed in Theorem~\ref{theo:class-S} and therefore, the above dynamics can be reverse-engineered as a primal-dual gradient algorithm to solve 
\begin{align} \label{eq:PI_f}
\begin{array}{l}
\mathop {\max }\limits_{z \in \mathbb{R}^{n-1} } \mathop {\min }\limits_{y \in \mathbb{R}^n } f 
 = \frac{\rho_2}{2}y^T Ly + \rho_1 z^T \tilde{L}^T y + \frac{\delta }{2}y^T y - y^T d - y^T\delta y(0).
 \end{array}
\end{align}
It is of interest to redesign the distributed PI controller~\eqref{eq:integrator_u} to improve system performance.

\subsection{Problem Setup}
Motivated by the above examples, the desiderata is: for dynamical systems belonging to Class-$\mathcal{O}$ and Class-$\mathcal{S}$, improve their performance (convergence rates/speeds and transient behaviors) through redesigning the existing protocols and controls while maintaining their original control structures, i.e., the amount and topology of input channels remain unchanged.

\section{Redesign Methodology} \label{se:redesign}
In this section, we propose the redesign methodology for linear dynamical systems in Class-$\mathcal{O}$ and Class-$\mathcal{S}$. In particular, we analyze the convergence rates of the original systems and the redesigned systems. 
The corresponding explicit forms of the extra dynamics in the redesigned systems are derived. For convenience, we only consider discrete-time cases here, and continuous-time cases will be discussed in a future paper.

\begin{defi} \label{defi:f(x)}
\textbf{Function class $\mathscr{F}_L$}: $f \in \mathscr{F}_L$ means that the function $f$: $\mathbb{R}^{n}\rightarrow\mathbb{R}$ is convex and has $L$-Lipschitz continuous gradient ($L \geq 0$), i.e., $\forall x,y \in \mathbb{R}^n$, we have 
\begin{align}
\left\| {\nabla f(x) - \nabla f(y)} \right\| \le L \left\| {x - y} \right\|. \label{eq:lipschitz}
\end{align}

\noindent\textbf{Function class $\mathscr{S}_{\mu,L}$}: $f \in \mathscr{S}_{\mu,L}$ means that $f$ is $\mu$-strongly convex ($\mu >0$) and has $L$-Lipschitz continuous gradient ($L \geq 0$), i.e., $\forall x,y \in \mathbb{R}^n$, we have \eqref{eq:lipschitz} and 
\begin{align}
f(y) \ge f(x) + \nabla f(x)^T (y - x) + \frac{\mu }{2}\left\| {y - x} \right\|^2. 
\end{align}
\end{defi}

\begin{rem}
If the function $f \in \mathscr{S}_{\mu,L}$, there exist constants $L$ and $\mu$ such that $\forall x,y \in \mathbb{R}^n$,
\begin{align*} 
\mu \le\frac{\| \nabla f(x) - \nabla f(y)\|}{\|x-y\|} \le L 
\end{align*}

Furthermore, if the function $f$ is twice-differentiable, then its second derivative satisfies 
\begin{align*}
\mu I_n \preceq \nabla^2 f(x) \preceq LI_n.
\end{align*}
\indent Therefore, through the above inequalities, parameters $L$ and $\mu$ can be obtained.
\end{rem}

\subsection{Systems in Class-$\mathcal{O}$}
\subsubsection{Convergence Rate of the Original Systems}
Any system~\eqref{eq:LTI-sys} in Class-$\mathcal{O}$ can be reverse-engineered as a gradient descent (GD) algorithm to solve an unconstrained convex optimization problem, so the convergence rate of system~\eqref{eq:LTI-sys} in Class-$\mathcal{O}$ follows that of a GD algorithm. Suppose the reverse-engineered optimization problem is
\begin{align}
 \min_{x \in \mathbb{R}^n } f\left( x \right) \label{eq:uncons}
\end{align}
where $x \in \mathbb{R}^n$ is the decision vector and $f$ is a convex, scalar differentiable function of $x$. Assume that a global minimizer of problem~\eqref{eq:uncons} exists and is denoted by $x^*$. 

The GD algorithm is the simplest algorithm to solve problem~\eqref{eq:uncons} in discrete-time~\cite{nesterov2018lectures}, as shown in Algorithm 1, where $k=0, 1, \dots, N$ is the time step. Theorems~\ref{theo:GD} and \ref{theo:GD_stro} summarize the convergence property of the GD algorithm/original system when the objective function $f\in \mathscr{F}_L$ or $f\in \mathscr{S}_{\mu,L}$. 
\begin{algorithm}
\caption{Gradient descent algorithm}
\label{alg:GD}
\textbf{Setting:} Choose appropriate positive step size $\varepsilon$, and $x_0 \in \mathbb{R}^n$.    
\begin{align} \label{eq:GD}
{x_{k + 1}} = {x_k}-\varepsilon\nabla f\left( {{x_k}} \right)
\end{align}
\end{algorithm}

\begin{theo}  \label{theo:GD}
For any system~\eqref{eq:LTI-sys} in Class-$\mathcal{O}$, if the objective function $f$ obtained via reverse-engineering belongs to $\mathscr{F}_L$ and $0<\varepsilon<2/L$, then the convergence rate of this system is given by (See Theorem 2.1.14 in \cite{nesterov2018lectures} for the proof) 
\begin{align}
f\left( {{x_k}} \right) - {f(x^*)} \le \frac{{2{{\left\| {{x_0} - {x^*}} \right\|}^2}}}{{\varepsilon k}} = O\left( {\frac{1}{{\varepsilon k}}} \right), \quad \forall k \in \mathbb{Z}^+. \nonumber
\end{align}
\end{theo}
 
\begin{theo}  \label{theo:GD_stro}
For any system~\eqref{eq:LTI-sys} in Class-$\mathcal{O}$, if the objective function $f$ obtained via reverse-engineering belongs to $\mathscr{S}_{\mu,L}$ and $0<\varepsilon \le \frac{2}{{\mu  + L}}$, then the convergence rate of this system is given by
\begin{align*}
\left\| {x_{k}  -  x^* } \right\| \le  (1 - \mu \varepsilon )^k \left\| {x_0  - x^* } \right\|=O({e^{ - k \mu\varepsilon}}), \quad \forall k \in \mathbb{Z}^+.
\end{align*}
\end{theo}

\begin{proof}
Theorem 2.1.15 in~\cite{nesterov2018lectures} and Theorem 3.12 in~\cite{Bubeck2014convex} showed that $\left\| {x_{k + 1}  -  x^* } \right\|^2$$\le$ $\left(1 - \frac{{2\varepsilon \mu L}}{{L + \mu}}\right)$\\$\times\left\| {x_k  -  x^* } \right\|^2$$+ \varepsilon \left (\varepsilon  - \frac{2}{{L + \mu }}\right)\left\| {\nabla f(x_k )} \right\|^2$. But Theorem~\ref{theo:GD_stro} actually provides a more strict upper bound compared with Theorem 2.1.15 in~\cite{nesterov2018lectures}. The step size should be in the range of $0<\varepsilon \le \frac{2}{{\mu  + L}}$, so that the sequence $\|x_k - x^*\|_{k\ge 0}$ is decreasing. Applying inequality $\left\|\nabla f(y) - \nabla f(x) \right\| \ge \mu \left\|y-x \right \|$ and optimality condition $\nabla f( x^* ) = 0$, we have 
\begin{align*}
{\left\| {{x_{k + 1}} - {{ x}^*}} \right\|^2} \le& \left( {1 - \frac{{2\varepsilon \mu L}}{{L + \mu}} + \varepsilon {\mu^2}\left( {\varepsilon  - \frac{2}{{L + \mu }}} \right)} \right){\left\| {{x_k} - {{ x}^*}} \right\|^2}
=\left (1- \mu \varepsilon \right)^2{\left\| {{x_k} - {{ x}^*}} \right\|^2}.
\end{align*}
The value of $\left(1 - \mu \varepsilon \right)$ is non-negative when $0<\varepsilon \le \frac{2}{{\mu  + L}}$ and $\mu \le L$. Removing squares on both sides of the inequality completes the proof. Note that the last equation in this theorem is obtained via the inequality $(1-t)\le e^{-t}, \forall t$.
\end{proof}

\begin{coro} \label{coro:GD}
For any system~\eqref{eq:LTI-sys} in Class-$\mathcal{O}$, if the objective function $f$ obtained via reverse-engineering belongs to $\mathscr{S}_{\mu,L}$ and the step size $\varepsilon=\frac{2}{{\mu  + L}}$, then this system is with the optimal convergence rate given by~\cite{nesterov2018lectures}
\begin{align}
\left\| {x_{k}  -  x^* } \right\| \le \left(\frac{{L  - \mu}}{{L  + \mu}}\right)^k \left\| {x_0  -  x^* } \right\| =  O({e^{ - k \frac{2\mu}{L+\mu}}}), \forall k \in \mathbb{Z}^+.   \nonumber
\end{align}
\end{coro}

\subsubsection{Redesign via a Heavy Ball Method}
Once reverse-engineering a given system~\eqref{eq:LTI-sys} in Class-$\mathcal{O}$ as a GD algorithm to solve~\eqref{eq:uncons}, it is natural to apply faster algorithms to solve this problem, which can result in a redesigned system formula with improved performance.


In Algorithm~\ref{alg:GD}, the next point only depends on the current point like in a Markov chain. The heavy ball (HB) method, introduced by Polyak~\cite{polyak1987introduction}, utilizes a momentum term $x_k-x_{k-1}$ to incorporate the effect of second-order change, often leading to smoother trajectories and a faster convergence rate~\cite{polyak1987introduction}. Algorithm~\ref{alg:HB} shows its iterations.

\begin{algorithm}
\caption{Heavy ball method}
\label{alg:HB}
\textbf{Setting:} Choose appropriate positive step size $\varepsilon$, coefficient $\beta$ and let $x_1 = x_0$ be an arbitrary initial condition.    
\begin{align}
{x_{k + 1}} = {x_k} - \varepsilon \nabla f({x_k}) + \beta ({x_k} - {x_{k - 1}}) \label{eq:HB}
\end{align}
\end{algorithm}

For any system~\eqref{eq:LTI-sys} in Class-$\mathcal{O}$, to improve system performance, apply the HB method to solve problem~\eqref{eq:uncons} to obtain
\begin{align}\label{eq:HB_uk}
    {x_{k + 1}} = {x_k} - \varepsilon \nabla f({x_k}) + \underbrace {\beta ({x_k} - {x_{k - 1}})}_{\Delta {u_k}} 
\end{align}
where $\Delta {u_k} = \beta ({x_k} - {x_{k - 1}})$ is the explicit form of extra dynamics. It is clear that~\eqref{eq:HB_uk} is the sum of the GD formula~\eqref{eq:GD} and extra dynamics $\Delta u_k$. 
Therefore, the redesigned system via the HB method is 
\begin{align}\label{eq:HB_re}
x_{k+1}=Ax_k+Cw+\Delta u_k
\end{align}
where $\Delta {u_k} = \beta ({x_k} - {x_{k - 1}})$. Theorem~\ref{theo:HB} summarizes the convergence property when applying the HB method to redesign.

\begin{theo} \label{theo:HB}
For any system~\eqref{eq:LTI-sys} in Class-$\mathcal{O}$, if the objective function $f$ obtained via reverse-engineering belongs to $\mathscr{S}_{\mu,L}$ and is twice-differentiable, the step size $\varepsilon=\frac{4}{{{{\left( {\sqrt L  + \sqrt \mu  } \right)}^2}}}$ and $\beta = \frac{\sqrt{L}-\sqrt{\mu}}{\sqrt{L}+\sqrt{\mu}}$, then the redesigned system via the HB method~(\ref{eq:HB_re}) is with the optimal convergence rate given by (See Section 3.2.1 in~\cite{polyak1987introduction} for the proof): 
\begin{align}
\left\| {x_k  -  x^* } \right\| \le \left(\frac{{\sqrt L- \sqrt{\mu}}}{{\sqrt L + \sqrt{\mu}}}\right)^k \left\| {x_0 -  x^* } \right\|=O({e^{-k\frac{2\sqrt{\mu}}{\sqrt L + \sqrt{\mu}}}}), \quad \forall k \in \mathbb{Z}^+.  \nonumber
\end{align}
\end{theo}

According to Theorem~\ref{theo:HB}, when implementing the HB method, constant coefficients can be adopted for simplification, i.e., $\varepsilon=\frac{4}{{{{\left( {\sqrt L  + \sqrt \mu  } \right)}^2}}}$, $\beta = \frac{\sqrt{L}-\sqrt{\mu}}{\sqrt{L}+\sqrt{\mu}}$.

\subsubsection{Redesign via an Accelerated Gradient Descent Method}
Similar to the HB method, an accelerated gradient descent (AGD) method also uses a momentum term. But it constructs an arbitrary auxiliary sequence $y_k$. 
AGD achieves the optimal convergence rate $O(1/\varepsilon k^2)$ when objectives belong to function class $\mathscr{F}_L$. Algorithm~\ref{alg:AGD} shows its iterations.
\begin{algorithm}
\caption{Accelerated gradient descent method}\label{alg:AGD}
\textbf{Setting:} Choose appropriate positive step size $\varepsilon$, and let $x_1 = x_0$ be an arbitrary initial condition.    
\begin{subequations} \label{eq:AGD}
\begin{align}
{x_{k + 1}} &= {y_k} - \varepsilon \nabla f\left( {{y_k}} \right) \label{AGD_1} \\
{y_k} &= {x_k} + \beta_k \left( {{x_k} - {x_{k - 1}}} \right) \label{AGD_2}
\end{align}
\end{subequations}
\end{algorithm}


For any system~\eqref{eq:LTI-sys} in Class-$\mathcal{O}$, apply the AGD method to solve problem~\eqref{eq:uncons} obtained via reverse-engineering to improve system performance. To show that the implementation is equivalent to introducing an extra part to the original GD dynamics,~\eqref{eq:AGD} is rearranged by interchanging the notations of $x$ and $y$ and combining the two equations:
\begin{align} \label{eq:AGD_combine}
    x_{k + 1}  = x_k  - \varepsilon \nabla f(x_k ) + \underbrace {\beta_k (y_{k + 1}  - y_k )}_{\Delta u_k }
\end{align}
where $\Delta u_k  = \beta_k (y_{k + 1}  - y_k ) = \beta_k \big[x_k  - \varepsilon \nabla f(x_k ) - x_{k - 1}  + \varepsilon \nabla f(x_{k - 1} )\big]$. 
It is clear that~\eqref{eq:AGD_combine} is the sum of the gradient descent formula~\eqref{eq:GD} and extra dynamics $\Delta u_k$. 
Therefore, the redesigned system via the AGD method is 
\begin{align}\label{eq:AGD_re}
x_{k+1}=Ax_k+Cw+\Delta u_k
\end{align}
where $\Delta u_k  =  \beta_k \big(Ax_k+Cw - Ax_{k-1}-Cw^{\prime} \big)$ and $w^{\prime}$ is the value of $w$ at time step $k-1$.  Theorems~\ref{theo:AGD} and~\ref{theo:AGD_stro} summarize the convergence property when applying the AGD method to redesign.

\begin{theo}\label{theo:AGD}
For any system~\eqref{eq:LTI-sys} in Class-$\mathcal{O}$, if the objective function $f$ obtained via reverse-engineering belongs to $\mathscr{F}_L$, the step size $0<\varepsilon \le 1/L$ and $\beta_k = \frac{{k - 1}}{{k + 2}}$, then the convergence rate of the redesigned system via the AGD method~(\ref{eq:AGD_re}) is (See Theorem 2.2.2 in \cite{nesterov2018lectures} for the proof)
\begin{align}
f\left( x_k \right) - {f(x^*)} \le \frac{{8\left\| {x_0  - x^* } \right\|^2 }}{3{\varepsilon (k + 1)^2 }} = O\left(\frac{1}{\varepsilon k^2}\right), \quad \forall k \in \mathbb{Z}^+. \nonumber
\end{align}
\end{theo}



\begin{theo}\label{theo:AGD_stro}
For any system~\eqref{eq:LTI-sys} in Class-$\mathcal{O}$, if the objective function $f$ obtained via reverse-engineering belongs to $\mathscr{S}_{\mu,L}$, the step size $\varepsilon=\frac{1}{L}$ and $\beta_k = \frac{\sqrt{L}-\sqrt{\mu}}{\sqrt{L}+\sqrt{\mu}}$ (constant step scheme III in~\cite{nesterov2018lectures}), then the redesigned system via the AGD method~(\ref{eq:AGD_re}) is with the optimal convergence rate~\cite{Ghadimi2015Global}
\begin{align}
\left\| {x_k  -  x^* } \right\| \le \left(1 - \sqrt{\frac{\mu}{L}}\right)^k \left\| {x_0 -  x^* } \right\| =O({e^{-k\frac{\sqrt{\mu}}{\sqrt L}}}), \quad \forall k \in \mathbb{Z}^+.\nonumber   
\end{align}
\end{theo}

According to Theorem~\ref{theo:AGD_stro}, when implementing the AGD method, constant coefficients can be adopted for simplification, i.e., $\varepsilon=\frac{1}{L}$, $\beta_k = \frac{\sqrt{L}-\sqrt{\mu}}{\sqrt{L}+\sqrt{\mu}}$ in Algorithm~\ref{alg:AGD} to achieve the optimal convergence rate.

Note that although both Theorem~\ref{theo:HB} and Theorem~\ref{theo:AGD_stro} require $f \in \mathscr{S}_{\mu,L}$ and the constant coefficient $\beta/\beta_k$ is related to $\mu$, our HB/AGD-based redesign is not restricted to strongly convex functions (obtained via reverse-engineering) only. It is still effective for $f \in \mathscr{F}_{L}$, but the convergence rates are unclear when coefficients are constants~\cite{Ghadimi2015Global}. In this case, we can adjust $\beta/\beta_k$ manually. Note that increasing the value of $\beta/\beta_k$ enlarges the influence of the momentum term $x_k-x_{k-1}$ in~\eqref{eq:HB_uk} and $y_{k+1}-y_{k}$ in~\eqref{eq:AGD_combine}. 

TABLE~\ref{table:conv} compares the convergence rates of the original and redesigned systems under HB and AGD methods.  For any system~\eqref{eq:LTI-sys} in Class-$\mathcal{O}$, 
When the objective function $f$ obtained via reverse-engineering is in $\mathscr{F}_L$, the redesigned system via the AGD method has a guaranteed better convergence rate than the original system. When $f \in \mathscr{S}_{\mu,L}$, the optimal convergence rate of the redesigned system via the HB method is always better than that of the original system and the redesigned system via the AGD method. 
\footnotetext[4]{For objectives in this class, we compare the optimal convergence rate for convenience.}

\begin{table}[hb]
\begin{center}
\caption{Convergence rate comparison.}\label{table:conv}
\begin{tabular}{c|c|c|c}

\hline \hline
$f$ Classes   & Original & Apply HB & Apply AGD  \\ 
\hline
$f \in \mathscr{F}_L$  & $O\left( {\frac{1}{{\varepsilon k}}} \right)$ & $\diagup$ & $O\left(\frac{1}{{\varepsilon k^2 }}\right)$   \\ \hline
$f \in \mathscr{S}_{\mu,L}$\footnotemark[4] & $O({e^{ - k \frac{2\mu}{L+\mu}}})$ & $O({e^{-k\frac{2\sqrt{\mu}}{\sqrt L + \sqrt{\mu}}}})$ &  $O({e^{-k\frac{\sqrt{\mu}}{\sqrt L}}})$ \\ \hline\hline
\end{tabular}
\end{center}
\end{table}

\subsection{Systems in Class-$\mathcal{S}$}
\subsubsection{Convergence Rate of the Original Systems}
Any system~\eqref{eq:LTI-sys} in Class-$\mathcal{S}$ can be reverse-engineered as a primal-dual gradient (PDG) algorithm to solve a convex-concave saddle-point problem, so the convergence rate of system~\eqref{eq:LTI-sys} or \eqref{equ:linearsystemsplit} follows that of a PDG algorithm. Specifically, we focus on $A_{11}=I_{n_1}$ in~\eqref{equ:linearsystemsplit} so that the optimization problem obtained via reverse-engineering is
\begin{align}
    \mathop {\max }\limits_{\lambda  \in \mathbb{R}^m}  \mathop {\min }\limits_{x \in \mathbb{R}^n} \mathcal{L}(x,\lambda ) = f(x) + {\lambda ^T}(Bx - b) \label{eq:dual}
\end{align}
where $x \in \mathbb{R}^n$, $\lambda \in \mathbb{R}^m$ is the Lagrangian multiplier vector (dual variable vector), $B \in \mathbb{R}^{m \times n}$ and $b \in  \mathbb{R}^{m}$. Here we have rearranged~(\ref{equ:classSopt}) as above and replaced $x^{(1)},x^{(2)}, n_1,n_2$ with $\lambda, x,m,n$. Notations $f,n$ are abused in contrast to $f,n$ in~(\ref{equ:classSopt}).

\begin{ass} \label{ass:rank_B}
Matrix B is full row rank.
\end{ass}

The corresponding primal problem of~\eqref{eq:dual} is an equality constrained convex optimization problem given by
\begin{align}
   \min_{x \in \mathbb{R}^n } f\left(x \right) \quad
 \text{s.t.} {\kern 1pt} {\kern 1pt} {\kern 1pt} {\kern 1pt} {\kern 1pt} {\kern 1pt} Bx = b. \label{eq:primal}   
\end{align}
A PDG algorithm is the simplest algorithm to solve problem~\eqref{eq:dual}, as shown in Algorithm~\ref{alg:PDG}.
\begin{algorithm}
\caption{First-order primal-dual gradient method}
\label{alg:PDG}
\textbf{Setting:} Choose appropriate positive step sizes $\varepsilon_1$, $\varepsilon_2$, and let $x_0$ and $\lambda_0$ be arbitrary initial conditions.   
\begin{subequations} \label{eq:PDG}
\begin{align}
{x_{k + 1}} &= {x_k} - {\varepsilon _1}(\nabla f(x_k) + {B^T}{\lambda _k}) \label{alg:PDG_1}\\
{\lambda _{k + 1}} &= {\lambda _k} + {\varepsilon _2}(B{x_k} - b) \label{alg:PDG_2}
\end{align}
\end{subequations}
\end{algorithm}

Assume that a saddle point of $\mathcal{L}$ exists and is denoted by $(x^*, \lambda^*)$, which satisfies the optimality conditions:
\begin{subequations} \label{eq:L_opt}
\begin{align}
\nabla_x {\mathcal{L}} &= \nabla f({x^*}) + {B^T}{\lambda ^*} = 0 \label{eq:optimality_x}\\
\nabla_{\lambda} {\mathcal{L}} &= Bx^*-b = 0.
\end{align}
\end{subequations}

Furthermore, assume the objective function obtained via reverse-engineering $f \in \mathscr{S}_{\mu,L}$. Since the primal problem is strongly convex and the constraint is affine, strong duality holds. Therefore, $x^*$ is unique and is an optimal solution of the primal problem~\eqref{eq:primal}. When Assumption~\ref{ass:rank_B} holds, $\lambda^*$ is unique.

Introduce the conjugate function of $f(x)$: $f^{\star} (y) = \mathop {\sup }\limits_{x \in \mathbb{R}^n} \left\{ {\left\langle {x,y} \right\rangle  - f(x)} \right\}$ for all $y \in \mathbb{R}^n$, and~\eqref{eq:dual} can be further expressed as 
$\mathop {\min }\limits_{\lambda \in \mathbb{R}^m}  {\rm{ }}{f^{\star}}( - {B^T}\lambda ) + {\lambda ^T}b$. Theorem~\ref{theo:PDG_conv} summarizes the convergence rate of the PDG algorithm/original system.

%

 \begin{theo} \label{theo:PDG_conv}
For any system~\eqref{eq:LTI-sys} in Class-$\mathcal{S}$, if the objective function $f$ in~(\ref{eq:dual}) obtained via reverse-engineering belongs to $\mathscr{S}_{\mu,L}$, Assumption~\ref{ass:rank_B} holds and step sizes $\varepsilon _1 , \varepsilon _2$ satisfy $0 < \varepsilon _1  \le \frac{2}{{L + \mu }}$, $0 < \varepsilon _2  \le \frac{2}{{\sigma _{\min }^2(B) /L + \sigma _{\max }^2(B) /\mu }}$ and $c< 1$, then this system converges to the unique saddle point $(x^*,\lambda^*)$ exponentially. Let $a_k=\left\| {x_k  - \nabla f^{\star} ( - B^T \lambda _k )} \right\|$, $b_k = \left\| {\lambda _k  - \lambda ^* } \right\|$ and define a potential function $V_k=\gamma a_k+b_k$, then for some constants $c, \gamma$ that depend on $\varepsilon_1, \varepsilon_2$, we have
\begin{align} \label{conv:PDG}
V_{k + 1}  \le c V_k, \quad  \forall k \in \mathbb{Z}^+.
\end{align}
where $\gamma>0$ and $c=\max \{ c_1 ,c_2 \}$ with $c_1  =  1 - \mu \varepsilon _1   + \frac{{\varepsilon _2 \sigma _{\max }^2(B) }}{\mu } + \frac{{\varepsilon _2 \sigma _{\max }(B) }}{\gamma }$ and $c_2  = 1 - \frac{{\varepsilon _2 \sigma _{\min }^2(B) }}{L} + \frac{{\varepsilon _2 \gamma \sigma _{\max }^3(B) }}{{\mu ^2 }}$.
\end{theo}

\begin{proof}
See the Appendix.
\end{proof}

Note that the potential function $V_k$ decreases at a geometric rate and the error of $\left\| {{x_k} - {x^*}} \right\|$ and $\left\| {{\lambda _k} - {\lambda ^*}} \right\|$ are bounded by $V_k$: $\left\| {{\lambda _k} - {\lambda ^*}} \right\| \le {V_k}$ and $\left\| {{x_k} - {x^*}} \right\| \le \left\| {{x_k} - \nabla {f^{\star}}( - {B^T}{\lambda _k})} \right\| + \left\| {\nabla {f^{\star}}( - {B^T}{\lambda _k}) - {x^*}} \right\|
 \le {a_k} + \frac{{{\sigma _{\max }(B)}}}{\mu }{b_k}
 \le \text{max}\left\{ {\frac{1}{\gamma },\frac{{{\sigma _{\max }(B)}}}{\mu }} \right\}{V_k}$. Therefore, as $V_k$ approaches zero, $(x_k,\lambda_k)$ approaches the saddle point $(x^*,\lambda^*)$.



\begin{coro} \label{coro:opt_stepsize}
When the step sizes $\varepsilon_1 = \frac{2}{L+\mu}$ and $\varepsilon_2 = \left( \frac{\sigma_{\max}^2(B)}{\mu} + \frac{\sigma_{\max}(B)}{\gamma}+\frac{\sigma_{\min}^2(B)}{L}  -\frac{\gamma \sigma_{\max}^3(B)}{\mu^2} \right)^{-1} \frac{2\mu}{L+\mu}$, the optimal convergence rate is attained for the original system.
\end{coro}

\begin{proof}
Since the geometric factor is determined by $c$, the convergence rate is optimal when $c$ is minimized. For a specific problem, parameters including $\mu, L, \gamma, \sigma_{\max}(B),$ $\sigma_{\min}(B)$ are fixed and we are able to adjust step sizes only. It is straightforward to notice that $c_1  =  1 - \mu \varepsilon _1   + \frac{{\varepsilon _2 \sigma _{\max }^2(B) }}{\mu } + \frac{{\varepsilon _2 \sigma _{\max }(B) }}{\gamma }$ is monotonically decreasing in $\varepsilon_1$ and increasing in $\varepsilon_2$ while $c_2  = 1 - \frac{{\varepsilon _2 \sigma _{\min }^2(B) }}{L} + \frac{{\varepsilon _2 \gamma \sigma _{\max }^3(B) }}{{\mu ^2 }}$ is monotonically decreasing in $\varepsilon_2$ since $\gamma < \frac{\mu^2 \sigma_{\min}^2(B)}{L \sigma_{\max}^3(B)}$ (due to $c<1$). Thus, $c$ is minimized when $\varepsilon _1$ takes its upper limit $\frac{2}{{L + \mu }}$ and $c_1 = c_2$, from which we obtain the value of $\varepsilon_2$ as given in this corollary.
\end{proof}

%

\begin{coro} \label{coro:opt_rate}
Suppose $\gamma {\rm{ = }}\frac{{\mu ^{\rm{2}} \sigma _{\min }^2 (B)}}{{2L\sigma _{\max }^3 (B)}} $ and step sizes are chosen as in Corollary~\ref{coro:opt_stepsize}, then we have $c\le 1 - \frac{1}{{\kappa ^3 (4\tau ^2 + 2\tau + 1 )}} $, where $\kappa = \frac{L}{\mu} $ is the condition number and $\tau = \frac{\sigma_{\max}^2(B)}{\sigma_{\min}^2(B)}$.
\end{coro}

\begin{proof}
First we show $\varepsilon _2$ satisfies the bound in Theorem~\ref{theo:PDG_conv}.
\begin{align*}
\varepsilon _2  &= 4\mu ^3 L\sigma _{\min }^2 (B)(\mu  + L)^{ - 1} \big(4L^2 \sigma _{\max }^4 (B) + \mu ^2 \sigma _{\min }^4 (B) + 2L\mu \sigma _{\min }^2 (B)\sigma _{\max }^2 (B)\big)^{ - 1}  \\ 
  &\le \frac{{4\mu ^3 L}}{{2\mu \big((4L^2  + 2L\mu )\sigma _{\max }^2 (B) + \mu ^2 \sigma _{\min }^2 (B)\big)}} \\ 
  &\le \frac{2}{{\frac{{6\sigma _{\max }^2 (B)}}{\mu } + \frac{{\sigma _{\min }^2 (B)}}{L}}} \\ 
  &\le \frac{2}{{\frac{{\sigma _{\max }^2 (B)}}{\mu } + \frac{{\sigma _{\min }^2 (B)}}{L}}}.
\end{align*}
Then we show the upper bound of $c$. According to Corollary~\ref{coro:opt_stepsize}, when $\gamma {\rm{ = }}\frac{{\mu ^{\rm{2}} \sigma _{\min }^2 (B)}}{{2L\sigma _{\max }^3 (B)}}$, we have $c_1  = c_2  = 1 - \frac{{\varepsilon _2 \sigma _{\min }^2 (B)}}{{2L}}$. Therefore,
\begin{align*}
c &=1 - 2\mu ^3 \sigma _{\min }^4 (B) (\mu  + L)^{ - 1} \big(4L^2 \sigma _{\max }^4 (B) + \mu ^2 \sigma _{\min }^4 (B) + 2L\mu \sigma _{\min }^2 (B)\sigma _{\max }^2 (B)\big)^{ - 1} \\
  &\le 1 - \frac{{\mu ^3 }}{{L(4L^2 \tau ^2  + \mu ^2  + 2L\mu \tau )}} \\ 
  &\le 1 - \frac{{\mu ^3 }}{{L^3 (4\tau ^2 + 2\tau + 1 )}} \\ 
  &\le 1 - \frac{1}{{\kappa ^3 (4\tau ^2 + 2\tau + 1  )}}.
\end{align*}
\end{proof}

Corollary~\ref{coro:opt_rate} shows that the optimal convergence rate of the original system is related to the condition number $\kappa$ and $\tau$. A large $\kappa$ implies a slow convergence rate since the iterations oscillate back and forth\cite{ortega2000iterative}. 
Thus, we utilize augmented Lagrangian (AL) and hat-x methods to change the convergence rate by changing the value of $\kappa$ of the objective. Applying those methods to solve problem~\eqref{eq:dual} obtained via reverse-engineering results in a redesigned system formula with improved performance.

\subsubsection{Redesign via an Augmented Lagrangian Method}
Adding the square of the equality constraints as penalty terms can change the condition number of the original problem while the optimal solution stays unchanged. This method is known as the AL method. 
The corresponding AL function for problem~\eqref{eq:primal} is
\begin{align}\label{eq:AL}
\mathcal{L}_a  = f(x) + \lambda ^T (Bx - b) + \frac{\alpha }{2}\left\| {Bx - b} \right\|^2
\end{align}
where $\alpha >0$ is a scalar. Equation~\eqref{eq:AL} is the Lagrangian for
\begin{align}\label{eq:AL_problem}
    \mathop {\min }\limits_{x \in \mathbb{R}^n} f(x) + \frac{\alpha }{2}\left\| {Bx - b} \right\|^2 \quad
 \text{s.t.} {\kern 1pt} {\kern 1pt} {\kern 1pt} {\kern 1pt} {\kern 1pt} {\kern 1pt} Bx = b
\end{align}
which has the same minimum and optimal solution as the original problem~\eqref{eq:primal}~\cite{bertsekas1997nonlinear}.

For any system~\eqref{eq:LTI-sys} in Class-$\mathcal{S}$, applying the AL method to solve the problem obtained via reverse-engineering leads to a redesigned system with improved performance: 
\begin{align*} 
x_{k + 1}  = x_k  - \varepsilon _1 \left( {\nabla f(x_k ) + B^T \lambda _k } \right)\underbrace { - \varepsilon _1 \alpha B^T (Bx_k  - b)}_{\Delta u_k }   
\end{align*}
where $\Delta u_k={ - \varepsilon _1 \alpha B^T (Bx_k  - b)}$. Note that the extra part is added to primal variables only. 

Let $g(x)=f(x) + \frac{\alpha }{2}\left\| {Bx - b} \right\|^2$ and we compare the condition numbers of $f$ and $g$. 
The Hessian matrix of $g(x)$ is $\bf{H}$$_g ={\bf{H}}_f + \alpha B^TB$ and $\lambda_{\min}({\bf{H}}_g)I_n \preceq {\bf{H}}_g \preceq \lambda_{\max}({\bf{H}}_g)I_n$. The condition number of $g(x)$, denoted by $\kappa_g$, is $\frac{\lambda_{\max}({\bf{H}}_g)}{\lambda_{\min}({\bf{H}}_g)}$. Denote by $\kappa_0$ the condition number of $f$ and $\kappa_0=L/\mu$. For a specific problem, we are able to obtain the numerical form of matrix ${\bf{H}}_g$. In this case, we can calculate $\kappa_g$ and $\kappa_0$ precisely. If $\kappa_g < \kappa_0$, applying the AL method is able to improve the convergence rate; otherwise, $\alpha$ should be set to zero to avoid the influence of penalty terms. Note that matrix $B$ has a significant influence on the effectiveness of this method. The matrix product $B^TB$ will not change the topological structure of the original system as the redesigned system can be rearranged as $x_{k + 1}  = x_k  - \varepsilon _1  \nabla f(x_k ) -\varepsilon _1B^T \big(\lambda _k  +\alpha (Bx_k  - b)\big)$ where $Bx_k  - b$ can be obtained from the original update of $\lambda$.

Another benefit of applying the AL method is the convexification when the objective $f$ obtained via reverse-engineering is not strongly convex in $\mathbb{R}^n$.

\begin{theo}~\label{pro:convexification}
Assume that ${\bf{H}}_f$ is positive definite on the nullspace of $B^TB$, i.e., $y^T {\bf{H}}_f y >0$ for all $y \neq 0$ with $y^TB^TBy=0$. Then there exists a scalar $\bar{\alpha}$ such that for $\alpha>\bar{\alpha}$, we have (See Theorem 4.2 in~\cite{anstreicher2000note} for the proof)
\begin{align*}
{\bf{H}}_f + \alpha B^TB \succ 0.
\end{align*}
\end{theo}


\subsubsection{Redesign via a Hat-x method}
Another way is to introduce a free variable $\hat x \in \mathbb{R}^n$ to prevent a dramatic change of $x$. This leads to
\begin{align} \label{eq:hat}
 \mathop {\min }\limits_{x, \hat{x} \in \mathbb{R}^n} f(x) + \frac{\alpha }{2}\left\| {x - \hat x} \right\|^2 \quad
 \text{s.t.} \quad Bx = b
\end{align}
where $\alpha$ is a scalar. This problem is equivalent to
\begin{align*}
  \mathop {\min }\limits_{z \in \mathbb{R}^{2n}} h(z) \qquad  \text{s.t.} \quad\bar Bz = b 
\end{align*}
where $z = \left[ {\begin{array}{*{20}c}
   x  \\
   {\hat x}  \\
\end{array}} \right]$, $h(z) = f(x) + \frac{\alpha }{2}z^T \left[ {\begin{array}{*{20}c}
   I_n & { - I_n}  \\
   { - I_n} & I_n \\
\end{array}} \right]z$, $\bar B = \left[ {\begin{array}{*{20}c}
   B & \mathbf{0}  \\
\end{array}} \right] \in \mathbb{R}^{m \times 2n} $. Suppose Assumption 1 holds, then $\bar B$ is also full row rank and the maximal and minimal singular values of $\bar B$ are the same as those of $B$. The Lagrangian is
\begin{align*}
\mathcal{L}_h  = f(x) + \lambda ^T (Bx - b) + \frac{\alpha }{2}\left\| {x - \hat x} \right\|^2.
\end{align*}
Applying an optimality condition, we have
\begin{subequations} \label{eq:L_h_opt}
\begin{align}
\nabla_x {\mathcal{L}_h} &= \nabla f({x^*}) + {B^T}{\lambda ^*} + \alpha (x^*- \hat x^* )= 0 \label{eq:optimality_hatx}\\
\nabla_{\hat x} {\mathcal{L}_h} &= \alpha (\hat x^* - x^*) = 0 \label{eq:optimality_hatx_2}\\
\nabla_{\lambda} {\mathcal{L}_h} &= Bx^*-b = 0.
\end{align}
\end{subequations} 
Then $\hat x^* = x^*$. 
Therefore, $x^*,\lambda^*$ in the optimal solution $(x^*,\hat x^*,\lambda^*)$ of~\eqref{eq:hat} are the same as that of~\eqref{eq:primal}. 

For any system~\eqref{eq:LTI-sys} in Class-$\mathcal{S}$, apply the hat-x method to solve the optimization problem obtained via reverse-engineering to improve system performance. The iterations of the primal variables are
\begin{subequations}
\begin{align*}
x_{k + 1}  &= x_k  - \varepsilon _1 \left( {\nabla f(x_k ) + A^T \lambda _k } \right)\underbrace { - \varepsilon _1 \alpha (x_k  - \hat x_k )}_{\Delta u_k } \\
 \hat x_{k + 1}  &= \hat x_k  + \varepsilon _1 \alpha (x_k  - \hat x_k )   
\end{align*}
\end{subequations}
where $\Delta u_k={ - \varepsilon _1 \alpha (x_k  - \hat x_k )}$. Note that the extra part is added to the primal variables only. 

Denote by $\kappa_h$ and ${\bf{H}}_h$ the condition number and the Hessian matrix of $h(z)$. Then 
$\kappa_h = \frac{\lambda_{\max}({\bf{H}}_h)}{\lambda_{\min}({\bf{H}}_h)}$. Next, we compare $\kappa_0$ and $\kappa_h$. To do this, we first obtain the range of $\kappa_h$, and then compare the lower and upper bounds of $\kappa_h$ with $\kappa_0$. 

\begin{lem} \label{lem:preceq}
For symmetric matrices $A, B \in \mathbb{R}^{n \times n}$, if $A \preceq B$, then $\lambda_{\max}(A) \le \lambda_{\max}(B)$ and $\lambda_{\min}(A) \le \lambda_{\min}(B)$ hold.
\end{lem}

\begin{proof}
For any $x \in \mathbb{R}^{n}$, we have $x^TAx \le x^TBx$. Assume $x^*=\mathop {\text{arg max}}\limits_{\|x\|=1} x^TAx$. Then $\lambda_{\max}(A)={x^*}^TAx^* \le {x^*}^TBx^* \le \mathop {\max}\limits_{\|x\|=1} x^TBx=\lambda_{\max}(B)$. On the other hand, assume $\bar x=\mathop {\text{arg min}}\limits_{\|x\|=1} x^TBx$. Then $\lambda_{\min}(B)= \bar x^TB \bar x \ge {\bar x}^TA\bar x \ge \mathop {\min}\limits_{\|x\|=1}  x^TAx = \lambda_{\min}(A)$.
\end{proof}

\begin{theo}
Assume the objective obtained via reverse-engineering in~\eqref{eq:primal} $f \in \mathscr{S}_{\mu,L}$, then $\underline{\kappa_h} \le  \kappa_h  \le \overline{\kappa_h}$, where $\underline{\kappa_h} = \frac{2\alpha+\mu+\sqrt{\mu^2+4\alpha^2}}{2\alpha+L-\sqrt{L^2+4\alpha^2}}$ and $\overline{\kappa_h}  = \frac{2\alpha+L+\sqrt{L^2+4\alpha^2}}{2\alpha+\mu-\sqrt{\mu^2+4\alpha^2}}$.
\end{theo}

\begin{proof}
Since $f \in \mathscr{S}_{\mu,L}$, we have $\mu I_n \preceq \nabla^2 f(x) \preceq LI_n$ and the Hessian of $h(z)$ is 
$H = \left[ {\begin{array}{*{20}c}
   \nabla^2 f(x)+\alpha I_n & -\alpha I_n  \\
   -\alpha I_n & \alpha I_n  \\
\end{array}} \right]
$. Then $\underline{H} \preceq  H \preceq \overline{H}$, where $\underline{H}=\left[ {\begin{array}{*{20}c}
   (\mu+\alpha) I_n & -\alpha I_n  \\
   -\alpha I_n & \alpha I_n  \\
\end{array}} \right]$ and $\overline{H}=\left[ {\begin{array}{*{20}c}
   (L+\alpha) I_n & -\alpha I_n  \\
   -\alpha I_n & \alpha I_n  \\
\end{array}} \right]$. Their eigenvalues are $\lambda (\underline{H})=\alpha+\frac{\mu}{2} \pm \frac{1}{2}\sqrt{\mu^2+4\alpha^2}$ and $\lambda(\overline{H})=\alpha+\frac{L}{2} \pm \frac{1}{2}\sqrt{L^2+4\alpha^2}$. By applying Lemma~\ref{lem:preceq}, we have $\lambda_{\min} (\underline{H}) \le \lambda_{\min}(H) \le \lambda_{\min} (\overline{H})$ and $\lambda_{\max} (\underline{H}) \le \lambda_{\max}(H) \le \lambda_{\max} (\overline{H})$. According to the definition of $\kappa_h$, we obtain the range of $\kappa_h$.
\end{proof}

Now we can compare the range of $\kappa_h$ with $\kappa_0$. Let $M=\underline{\kappa_h}-\kappa_0$, $w=\frac{\mu}{L}$ and $v=\frac{\alpha}{L}$ with $0<w \le 1$ and $v>0$, then $M =\frac{2v+w+\sqrt{w^2+4v^2}}{2v+1-\sqrt{1+4v^2}}-\frac{1}{w}.$
When $w$ tends to $0$, $M<0$; otherwise, $M>0$. On the other hand, $\overline{\kappa_h}-\kappa_0=\frac{1}{4\alpha \mu}\big[(2\alpha+\sqrt{L^2+4\alpha^2})(2\alpha+\mu+\sqrt{\mu^2+4\alpha^2})+L(\mu+\sqrt{\mu^2+4\alpha^2}-2\alpha)\big]$.
Since $\mu+\sqrt{\mu^2+4\alpha^2}>2\alpha$, $\overline{\kappa_h} > \kappa_0$. To conclude, $\kappa_h$ can be smaller than $\kappa_0$, depending on specific problems. Moreover, the hat-x method increases the dimension of the original system and works like a low-pass filter. This method could slow down system dynamics and smooth the trajectories, as demonstrated in Section~\ref{se:app_PI}.
\subsection{Reconfiguration Steps}



To summarize, the redesign approach is consist of the following three steps:

\begin{enumerate}
    \item \textbf{Reverse-engineering:} For a given system~\eqref{eq:LTI-sys}, apply Theorem~\ref{theo:class-O} or~\ref{theo:class-S} to reverse-engineer it as a GD algorithm~\eqref{eq:GD} or a PDG algorithm~\eqref{eq:PDG} for solving an unconstrained convex optimization problem~\eqref{eq:uncons} or a saddle-point problem~\eqref{eq:dual}, i.e., system dynamics~\eqref{eq:LTI-sys} can be rewritten as the form~\eqref{eq:GD} or~\eqref{eq:PDG}.

    \item \textbf{Acceleration:} Apply an HB/AGD method or an AL/hat-x method to solve the optimization problem obtained in Step 1), which results in a redesigned system.

    \item \textbf{Implementation:} Rearrange the redesigned system and compare it with the original system~\eqref{eq:GD}/\eqref{eq:PDG} or~\eqref{eq:LTI-sys} to obtain the implementation of the extra dynamics.
\end{enumerate}

This approach can be implemented by any dynamical system belonging to Class-$\mathcal{O}$ or Class-$\mathcal{S}$ to systematically improve its performance. With the extra dynamics added to the original controllers, the convergence rates/speeds and transient behaviors will be improved while the original control structures remain. Also, this approach is able to handle network systems as shown in Section~\ref{se:simu}, whose scale can be large.


\section{Examples Revisited}\label{se:simu}

\subsection{Internet Congestion Control}
Following the redesign procedure, applying Algorithm~\ref{alg:HB} and Algorithm~\ref{alg:AGD} to solve problem~\eqref{eq:opt_net}, we obtain the same form of redesigned dynamics but with different $\Delta u(k)$ given by
\begin{align*}
x(k+1)&=x(k) + \varepsilon \text{diag}\{k_{x_i}\} \left( U^{\prime} (x(k) ) - R^Tf(Rx(k)) \right) +\Delta u(k).
\end{align*}
For HB-based redesign, $\Delta u(k) = \beta \big ( x(k)-x(k-1)\big)$; while for AGD-based redesign, $\Delta u(k) = \beta_k \Big\{x (k) -x (k-1) + \varepsilon \text{diag}\{k_{x_i}\} \big[U ^\prime  (x (k)) - q(k)-U ^\prime (x (k-1)) + q(k-1) \big] \Big\}$, where $q$ is the corresponding vector form of $q_i$ and $q=R^Tf(Rx)$.


Consider a $2$-link network 
as shown in Fig.~\ref{fig:case}. Choose utility functions as $U_i(x_i)=\log x_i$, $k_{x_i}=0.1$, $\beta =0.54$ and $\beta_k =0.6$. Choose the penalty function as in~\cite{kelly1998rate}: $f_l (y_l ) = \frac{{(y_l  - c_l  + \sigma)^ +  }}{{\sigma ^2 }}$, where $\sigma $ is a small positive number and $(y_l  - c_l  + \sigma )^ + =\max\{y_l  - c_l  + \sigma,0\}$. Initially, the capacities of links A and B are $2, 4$ respectively and they change to $3, 1$ after some time to simulate sudden change of the link capacity. As shown in Fig.~\ref{fig:ICC}, redesigned dynamics via HB and AGD methods perform better than the original dynamics. 

\begin{figure}[!t]
\begin{center}
\includegraphics[width=10cm]{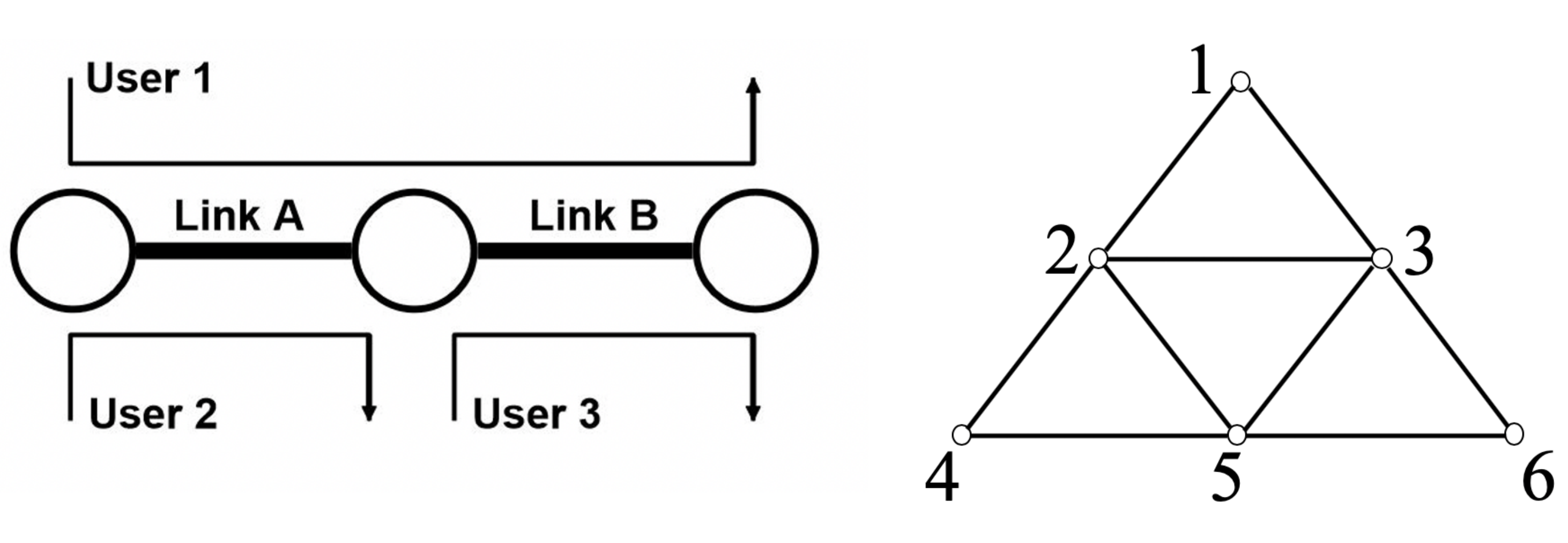}   
\caption{Left: A $2$-link network shared by $3$ users. Right: A regular network of $6$ agents.} 
\label{fig:case}
\end{center}
\end{figure}

\begin{figure*}[!t]
\centering
\begin{subfigure}[!t]{0.7\textwidth}
\centering
\includegraphics[width=1.1\textwidth]{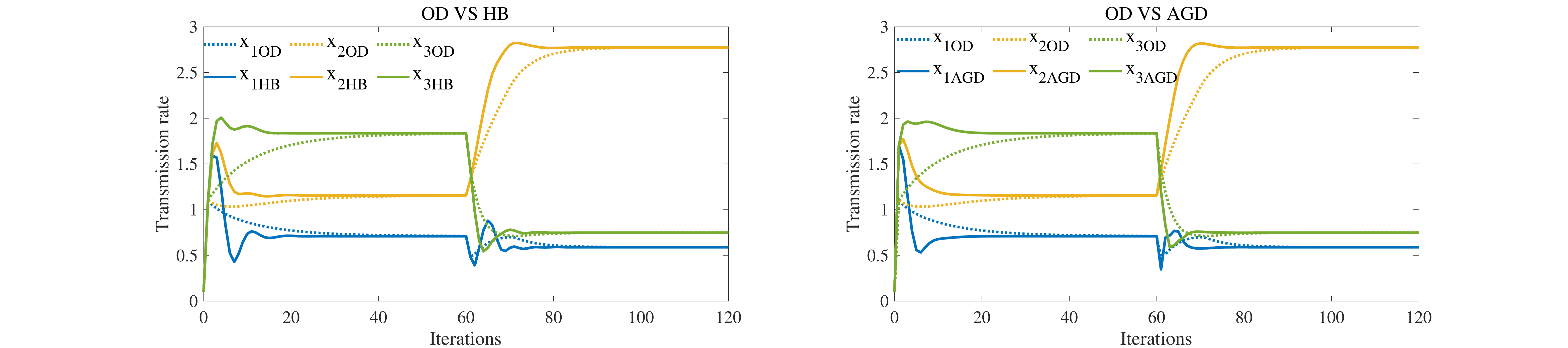}
\end{subfigure}
\begin{subfigure}[!t]{0.29\textwidth}
\centering
\includegraphics[width=1.1\textwidth]{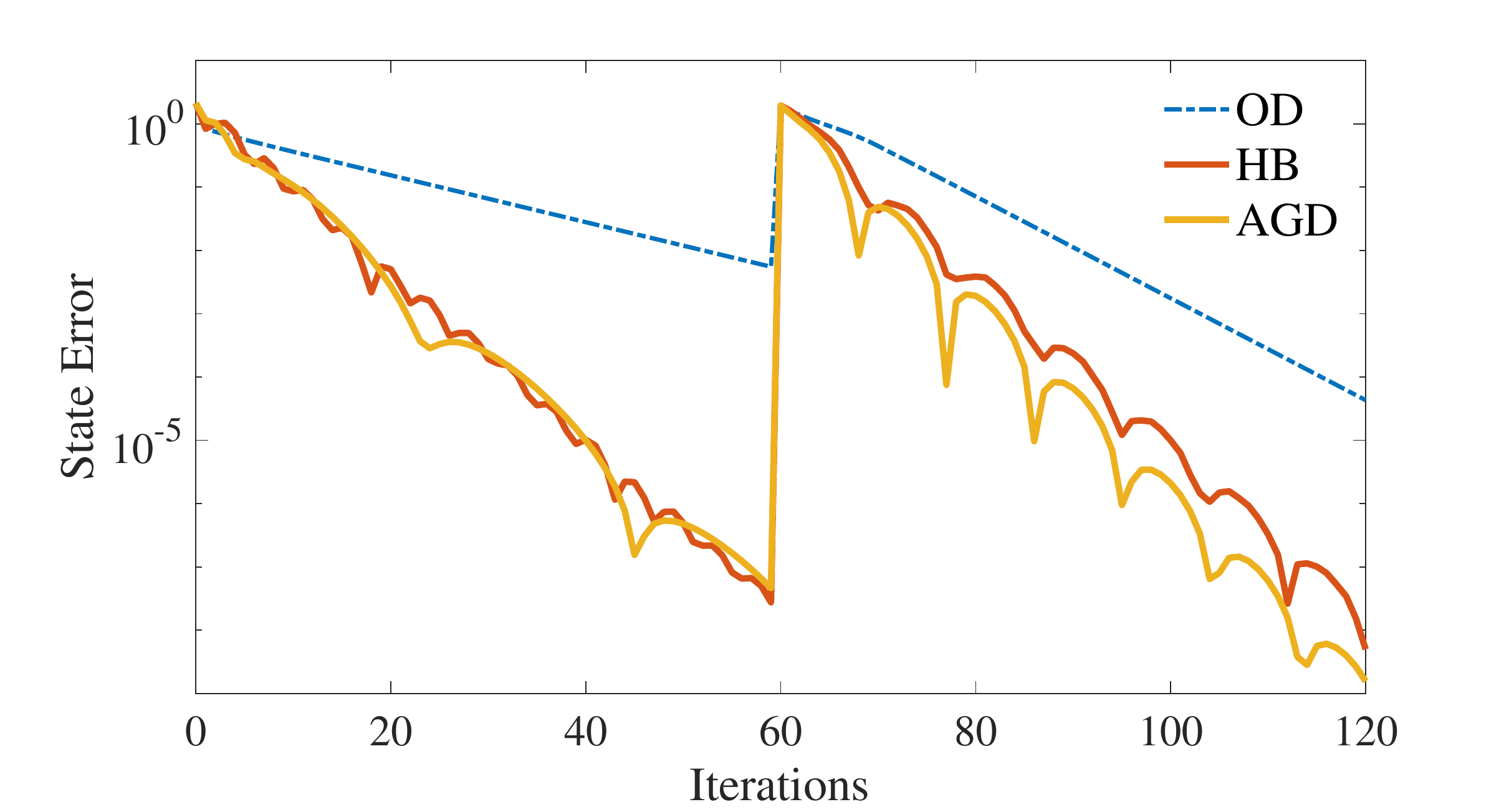}
\end{subfigure}
\caption{Internet congestion control. Left \& Middle: Simulation results (``OD'' stands for original dynamics; ``HB" and ``AGD" stand for redesigned dynamics via HB and AGD). Right: State error measured by $\left\| {{\rm{ }} x-x^* {\rm{ }}} \right\|_2$.}
\label{fig:ICC}
\end{figure*}

\begin{figure*}[!t]
\centering
\begin{subfigure}[!t]{0.7\textwidth}
\centering
\includegraphics[width=1.05\textwidth]{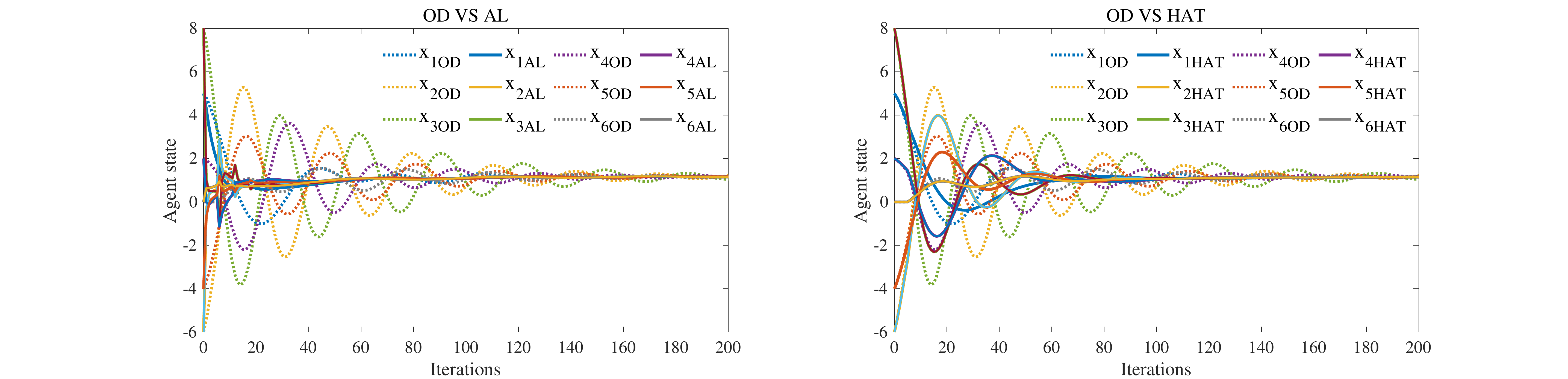}
\end{subfigure}
\begin{subfigure}[!t]{0.29\textwidth}
\centering
\includegraphics[width=1.05\textwidth]{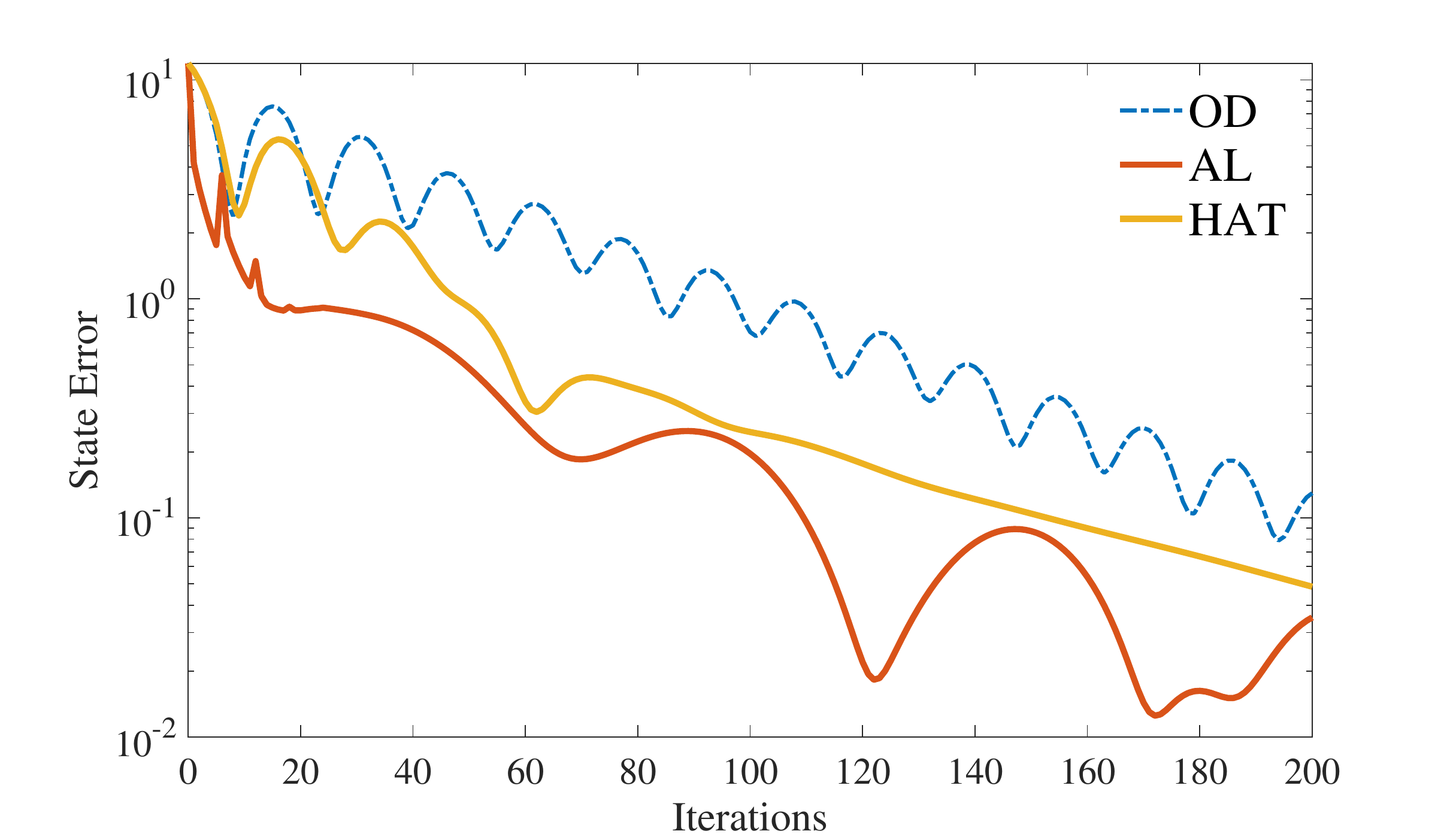}
\end{subfigure}
\caption{Distributed PI control. Left \& Middle: Simulation results (``OD'' stands for original dynamics; ``AL" and ``HAT" stand for redesigned dynamics via AL and hat-x). Right: State error measured by $\left\| {{\rm{ }} y-y^* {\rm{ }}} \right\|_2$.}
\label{fig:PI}
\end{figure*}

\subsection{Distributed PI Control for Single Integrator Dynamics}\label{se:app_PI}
Problem~\eqref{eq:PI_f} obtained via reverse-engineering is equal to 
\begin{align*}
\mathop {\min }\limits_{y \in \mathbb{R}^n } &f= \frac{\rho_2}{2}y^T Ly + \frac{\delta }{2}y^T y - y^T d - \delta y^T y(0) \\
&\text{s.t.} \quad \rho_1 \tilde{L}^T y = \textbf{0}. 
\end{align*}
Since ${\bf{H}}_f = \rho_2L+\delta I_n \succ 0$, $f \in \mathscr{S}_{L,u}$. Following the redesign procedure, applying augmented Lagrangian and hat-x method, we obtain the same form of redesigned dynamics but with different $\Delta u(k)$ given by
\begin{align*}
 y(k + 1) =& y(k) - \varepsilon _1 \big(\rho_2 Ly(k) + \delta y(k) - d - \delta y(0)  + \rho_1\tilde{L} z(k)\big)+ \Delta u(k) \\ 
 z(k+1) =& z(k)+ \varepsilon _{\rm{2}} D y(k).
\end{align*}
For the AL-based redesign, $\Delta u(k) =  - \varepsilon _1 \alpha \tilde{L} \tilde{L}^T y(k)$; while for the hat-x-based redesign, $\Delta u(k) =  - \varepsilon _1 \alpha \big(y(k) - \hat y(k)\big)$, where $\hat y(k + 1) = \hat y(k) + \varepsilon _1 \alpha (y(k) - \hat y(k) \big)$. 

Consider a network of $6$ agents with the topology shown in Fig.~\ref{fig:case}. In addition, communication delay is considered in this network since this could happen in reality. 
Let the constant disturbance $d=[0, 2, 0, 0, 0, 0]^T$, the initial condition $x(0)=[5,-6, 8, 2, -4, 0]^T$, the integral gain $\rho_1=10$, the static gain $\rho_2=0.5$, $\delta =1$. As shown in Fig.~\ref{fig:PI}, redesigned dynamics via the AL and hat-x methods perform better than the original dynamics, while the redesigned dynamics via hat-x is the most smooth. 

%
%


\section{Conclusion and Future work}\label{se:conclusion}
This paper has proposed a control reconfiguration approach to improve the performance of two classes of dynamical systems. This approach is to firstly reverse-engineer a given dynamical system as a gradient descent or a primal-dual gradient algorithm to solve a convex optimization problem. Then, by utilizing several acceleration techniques, the extra control term is obtained and added to the original control structure. Under this retrofit procedure, system performance is improved while the control structure remains, as demonstrated by both theoretical results and practical applications. 

In the future, we will investigate the implementation of the redesign in a continuous-time setting as well as the framework for nonlinear systems. Also, the case when ${\bf{H}}_{f(x^{(1)})} \preceq 0$ in the definition of Class-$\mathcal{S}$ will be studied. Last but not least, we will consider redesigning to improve other properties of systems, for example, robustness to delays.



\appendix[Proof of Theorem~\ref{theo:PDG_conv}]
This proof is inspired by~\cite{du2019linear}. 
According to~\cite{rockafellar1970convex,kakade2009duality}, if $f(x)\in \mathscr{S}_{L,u}$, then its conjugate function $f^{\star}(y)$ is $\frac{1}{L}$ strongly convex and has Lipschitz gradient $\frac{1}{\mu}$ in terms of $y$. Let $g(\lambda) = {\rm{ }}{f^{\star}}( - {B^T}\lambda ) + {\lambda ^T}b$, which is equivalent to~\eqref{eq:dual}. Proposition~\ref{prop:g_lambda} shows the strongly convex and Lipschitz continuous gradient parameters of $g(\lambda)$.

\begin{prop}\label{prop:g_lambda}
Function $g(\lambda)$ is $\frac{{\sigma _{\min }^2(B) }}{{L}}$strongly convex and has Lipschitz continuous gradient $\frac{{\sigma _{\max }^2(B) }}{\mu}$.
\end{prop}

\noindent \textit{Proof.} We prove this by applying the definition of strongly convex and Lipschitz continuous gradient. Choose any $\lambda_1, \lambda_2 \in \mathbb{R}^{m}$, we have
\begin{align*}
  \left\| {\nabla g(\lambda _1 ) - \nabla g(\lambda _2 )} \right\| &\le \left\| { - B\nabla f^{\star} ( - B^T \lambda _1 ) + B\nabla f^{\star} ( - B^T \lambda _2 )} \right\| \\ 
  &\le \sigma _{\max }(B) \left\| {\nabla f^{\star} ( - B^T \lambda _1 ) - \nabla f^{\star} ( - B^T \lambda _2 )} \right\| \\ 
  &\le \frac{{\sigma _{\max }(B) }}{\mu }\left\| { - B^T \lambda _1  - ( - B^T \lambda _2 )} \right\|
\le \frac{{\sigma _{\max }^2(B) }}{\mu }\left\| {\lambda _1  - \lambda _2 } \right\|.
\end{align*}

Therefore, $g(\lambda)$ has Lipschitz continuous gradient $\frac{{\sigma _{\max }^2(B) }}{\mu}$. On the other hand, for any $\lambda_1, \lambda_2 \in \mathbb{R}^{m}$, we have
\begin{align*}
g(\lambda _2 ) - g(\lambda _1 ) &= f^{\star} ( - B^T \lambda _2 ) + \lambda _2^T b - f^{\star} ( - B^T \lambda _1 ) - \lambda _1^T b \\ 
  &\ge \left\langle { \nabla f^{\star} ( - B^T \lambda _1 ), - B^T \lambda _2  + B^T \lambda _1 } \right\rangle  + (\lambda _2^T  - \lambda _1^T )b+ \frac{1}{{2L}}\left\| { - B^T \lambda _2  + B^T \lambda _1 } \right\|^2 \\
  &\ge \left\langle {\nabla g(\lambda _1 ),\lambda _2  - \lambda _1 } \right\rangle  + \frac{{\sigma _{\min }^2(B) }}{{2L}}\left\| {\lambda _2  - \lambda _1 } \right\|^2.
 \end{align*}
 Thus, $g(\lambda)$ is $\frac{{\sigma _{\min }^2(B) }}{L}$-strongly convex.
\\\\
Next, we establish the decrease of error term $\left\| {\lambda _k  - \lambda ^* } \right\|$ and $\left\| {x_{k} - \nabla f^{\star} ( - B^T \lambda _{k} )} \right\|$.

\begin{prop} \label{prop:b_t+1}
If $0 < \varepsilon _2  \le \frac{2}{{\sigma _{\min }^2(B) /L + \sigma _{\max }^2(B) /\mu }}$, then
\begin{align}
\left\| {\lambda _{k + 1}  - \lambda ^* } \right\| &\le \left(1 - \frac{{ \varepsilon _2\sigma _{\min }^2(B) }}{L}\right)\left\| {\lambda _k  - \lambda ^* } \right\| \nonumber + \varepsilon _2 \sigma _{\max }(B) \left\| {x_k -\nabla f^{\star} ( - B^T \lambda _k )  } \right\|.  
\end{align}
\end{prop}

\noindent \textit{Proof.} Define an auxiliary variable $\tilde \lambda _{k + 1}$:
\begin{align} \label{eq:tilde_lamdba}
\tilde \lambda _{k + 1}  &= \lambda _k  - \varepsilon _2 \nabla g(\lambda _k) \nonumber \\ 
  &= \lambda _k  - \varepsilon _2 \left( - B\nabla f^{\star} ( - B^T \lambda _k ) + b\right). 
\end{align}
Equation~\eqref{eq:tilde_lamdba} is a gradient descent step for the unconstrained problem $\mathop {\min }\limits_{\lambda \in \mathbb{R}^m}  {\rm{ }}{g(\lambda)}$. According to Theorem~\ref{theo:GD_stro} and Proposition \ref{prop:g_lambda}:
\begin{align}
\left\| {\tilde \lambda _{k + 1}  - \lambda ^* } \right\| \le \left(1 - \frac{\varepsilon _2 {\sigma _{\min }^2(B) }}{L}\right)\left\| {\lambda _k  - \lambda ^* } \right\|.     
\end{align}

On the other hand,
\begin{align*}
 \tilde \lambda _{k + 1}  - \lambda _{k + 1}  =& \lambda _k  - \varepsilon _2 \left( - B\nabla f^{\star} ( - B^T \lambda _k ) + b\right) - \lambda _k  + \varepsilon _2 ( - Bx_k  + b) \\
  =& \varepsilon _2 B\left( \nabla f^{\star} ( - B^T \lambda _k ) -x_k  \right). 
\end{align*}

Therefore, 
\begin{align*}
 \left\| {\lambda _{k + 1}  - \lambda ^* } \right\| \le& \left\| {\tilde \lambda _{k + 1}  - \lambda _{k + 1} } \right\| + \left\| {\tilde \lambda _{k + 1}  - \lambda ^* } \right\| \\ 
 \le& \left(1 - \frac{{ \varepsilon _2\sigma _{\min }^2(B) }}{L}\right)\left\| {\lambda _k  - \lambda ^* } \right\| + \varepsilon _2 \sigma _{\max }(B) \left\| {x_k - \nabla f^{\star} ( - B^T \lambda _k )  } \right\|.
\end{align*}

\begin{lem} \label{lem:tilde_x*}
In~\eqref{eq:primal}, assume $f \in \mathscr{S}_{\mu,L}$, let $\tilde f(x)=f(x) + x^T B^T \lambda _k$ and $\tilde x^*  = \arg \mathop {\min }\limits_{x \in {\mathbb{R}^{n}}} \tilde f(x)$, then $\tilde x^*   = \nabla f^{\star} ( - B^T \lambda _k)$ and as $k \to \infty$, $\tilde x^*$ tends to $x^*$.
\end{lem}

\noindent \textit{Proof.} For fixed $\lambda_k$, the update rule of $x_k$~\eqref{alg:PDG_1}: $x_{k + 1}  = x_k  - \varepsilon _1 (\nabla f(x_k ) + B^T \lambda _k)$ is a gradient descent step for the unconstrained problem $\mathop {\min }\limits_{x \in {\mathbb{R}^{n}}} \tilde f(x)$. Function $\tilde f(x)$ has the same function parameters as $f(x)$, i.e., $\tilde f(x)$ is also $\mu$-strongly convex and has Lipschitz continuous gradient $L$. According to the optimality condition, we have $\nabla \tilde f(\tilde x^* ) = \nabla f(\tilde x^* ) + B^T \lambda _k  = 0$. Since the gradient $\nabla f^{\star}$ is the inverse of $\nabla f$~\cite{rockafellar1970convex}, we have $\tilde x^*  = \nabla f^{ - 1} ( - B^T \lambda _k ) = \nabla f^{\star} ( - B^T \lambda _k)$. Similarly, according to~\eqref{eq:optimality_x}, we have $ x^* = \nabla f^{\star} ( - B^T \lambda^*)$. Since the sequence $\{(x_k,\lambda_k)\}_{k \ge  0}$ generated by Algorithm~\ref{alg:PDG} converges to $(x^*,\lambda^*)$~\cite{bertsekas1997nonlinear}, $\tilde x^*$ tends to $x^*$.

\begin{prop} \label{prop:a_t+1}
If $0 < \varepsilon _1  \le \frac{2}{{L + \mu }}$, then $\left\| {x{}_{k + 1} - \nabla f^{\star} ( - B^T \lambda _{k + 1} )} \right\| \le $
\begin{align*}
\qquad  \frac{{ \varepsilon _2 \sigma _{\max }^3(B)}}{{\mu ^2 }}\left\| {\lambda _k  - \lambda ^* } \right\| +\left(1 - \mu \varepsilon _1  + \frac{{\sigma _{\max }^2(B) \varepsilon _2 }}{\mu }\right)\left\| {x_k  - \nabla f^{\star} ( - B^T \lambda _k )} \right\|.
\end{align*}
\end{prop}

\noindent \textit{Proof.} Using Lemma~\ref{lem:tilde_x*} and Theorem~\ref{theo:GD_stro}, if $0 < \varepsilon _1  \le \frac{2}{{L + \mu }}$, we have
\begin{align*}
\left\| {x_{k + 1}  - \nabla f^{\star} ( - B^T \lambda _k )} \right\| \le \left( 1 - \mu \varepsilon _1 \right) \left\| {x_k  - \nabla f^{\star} ( - B^T \lambda _k )} \right\|.
\end{align*}

According to the update rule of $\lambda_k$~\eqref{alg:PDG_2}, we have 
\begin{align*}
\left\| {\lambda _{k + 1}  - \lambda _k } \right\| &= \varepsilon _2 \left\| {b - Bx_k } \right\| \\ 
&\le \varepsilon _2 \left\| {b - B\nabla f^{\star} ( - B^T \lambda _k )} \right\| + \varepsilon _2 \left\| {B\left(\nabla f^{\star} ( - B^T \lambda _k ) - x_k \right)} \right\| \\ 
&\le \varepsilon _2 \left\| {\nabla g(\lambda _k ) - \nabla g(\lambda ^* )} \right\| + \varepsilon _2 \sigma _{\max }(B) \left\| {\nabla f^{\star} ( - B^T \lambda _k ) - x_k } \right\| \\ 
&\le \frac{{\varepsilon _2 \sigma _{\max }^2(B) }}{\mu }\left\| {\lambda _k  - \lambda ^* } \right\| + \varepsilon _2 \sigma _{\max }(B) \left\| {x_k - \nabla f^{\star} ( - B^T \lambda _k ) } \right\|. 
\end{align*}
Using Proposition~\ref{prop:b_t+1} and the inequality above, we have 
\begin{align*}
&\left\| {x{}_{k + 1} - \nabla f^{\star} ( - B^T \lambda _{k + 1} )} \right\| \\
\le &\left\| {x{}_{k + 1} - \nabla f^{\star} ( - B^T \lambda _k )} \right\| + \left\| {\nabla f^{\star} ( - B^T \lambda _{k + 1} ) - \nabla f^{\star} ( - B^T \lambda _k )} \right\| \\ 
  \le &\left( 1 - \mu \varepsilon _1 \right) \left\| {x_k  - \nabla f^{\star} ( - B^T \lambda _k )} \right\| + \frac{{\sigma _{\max }(B) }}{\mu }\left\| {\lambda _{k + 1}  - \lambda _k } \right\| \\ 
  \le &\left(1 - \mu \varepsilon _1  + \frac{{\sigma _{\max }^2(B) \varepsilon _2 }}{\mu }\right)\left\| {x_k  - \nabla f^{\star} ( - B^T \lambda _k )} \right\| + \frac{{\varepsilon _2 \sigma _{\max }^3(B)  }}{{\mu ^2 }}\left\| {\lambda _k  - \lambda ^*} \right\|.    
\end{align*}
Note that $\varepsilon_2$ should be chosen relatively small to make sure $1 - \mu \varepsilon _1  + \frac{{\sigma _{\max }^2(B) \varepsilon _2 }}{\mu } < 1$ so that the sequence $\left(\left\| {x_k  - \nabla f^{\star} ( - B^T \lambda _k )} \right\|\right)_{k \ge 0}$ is decreasing. 

Finally, we use a potential function $V(k)$ to add the error terms. Using Proposition~\ref{prop:b_t+1} and Proposition~\ref{prop:a_t+1}, we have 
\begin{align*}
V_{k + 1}  =& \gamma \left\| {x_{k + 1}  - \nabla f^{\star} ( - B^T \lambda _{k + 1} )} \right\| + \left\| {\lambda _{k + 1}  - \lambda ^* } \right\| \\ 
\le& \left(1 - \mu \varepsilon _1  + \frac{{\varepsilon _2 \sigma _{\max }^2(B) }}{\mu }\right)\gamma a_k + \frac{{\varepsilon _2 \sigma _{\max }^3(B) \gamma }}{{\mu ^2 }}b_k + \left(1 - \frac{{\sigma _{\min }^2(B) \varepsilon _2 }}{L}\right)b_k + \varepsilon _2 \sigma _{\max }(B) a_k \\ 
\le& \left(1 - \mu \varepsilon _1  + \frac{{\varepsilon _2 \sigma _{\max }^2(B) }}{\mu } + \frac{{\varepsilon _2 \sigma _{\max }(B) }}{\gamma }\right)\gamma a_k + \left(1 - \frac{{\varepsilon _2 \sigma _{\min }^2(B) }}{L} + \frac{{\varepsilon _2 \gamma \sigma _{\max }^3(B) }}{{\mu ^2 }}\right)b_k\\
\le& cV_k. 
\end{align*}


Note that there is an upper limit for $\gamma$, i.e., $\gamma < \frac{\mu^2 \sigma_{\min}^2(B)}{L \sigma_{\max}^3(B)}$. We can choose large $\gamma$ and $\varepsilon_1$ (approaching their upper limits) and small $\varepsilon_2$ to make sure $c_1, c_2<1$ holds. 


\ifCLASSOPTIONcaptionsoff
  \newpage
\fi

\end{spacing}
\end{document}